\documentclass{lmcs}
\pdfoutput=1

\usepackage{lastpage}
\lmcsdoi{20}{3}{6}
\lmcsheading{}{\pageref{LastPage}}{}{}%
{Oct.~10,~2023}{Jul.~17,~2024}{}

\usepackage[utf8]{inputenc}

\keywords{coalgebraic logic, many-valued logic, modal logic, semi-primal algebras, one-step completeness, expressivity}

\usepackage{amsmath,amsthm,amssymb,tikz}
\usepackage{graphicx}
\usepackage[all,2cell]{xy}
\usepackage{tikz-cd}
\usepackage{hyperref}

\newcommand{\alg}[1]{\mathbf{#1}}
\newcommand{\cate}[1]{{\mathcal{#1}}}	
\newcommand{\var}[1]{\mathcal{#1}}
\newcommand{\func}[1]{\mathsf{#1}}
\newcommand{\lucas}{\text{\bf\L}}
\newcommand{\Stone}{\mathsf{Stone}}
\newcommand{\Set}{\mathsf{Set}}
\newcommand{\BA}{\mathsf{BA}}
\newcommand{\FLew}{\mathsf{FL_{ew}}}
\newcommand{\Krip}{\mathsf{Krip}}
\newcommand{\variety}[1]{\mathbb{H}\mathbb{S}\mathbb{P}(#1)}
\newcommand{\StoneD}{\mathsf{Stone}_{\alg{D}}}
\newcommand{\SetD}{\mathsf{Set}_{\alg{D}}}
\newcommand{\val}[1]{\mathbf{#1}}
\newcommand{\U}{\mathsf{U}}
\newcommand{\VS}{\mathsf{V}^\alg{S}}
\newcommand{\Vtop}{\mathsf{V}^\top}
\newcommand{\CS}{\mathsf{C}^\alg{S}}
\newcommand{\PS}{\Pow_\alg{S}}
\newcommand{\KS}{\mathsf{K}_\alg{S}}
\newcommand{\C}{\mathsf{C}}
\newcommand{\Pow}{\mathfrak{P}}
\newcommand{\Skel}{\mathfrak{S}}

\newcommand{\V}{\mathsf{V}}

\newcommand{\im}{\mathsf{im}}
\newcommand{\theory}{\mathsf{th}}
\newcommand{\id}{\mathsf{id}}
\newcommand{\Alg}[1]{\mathsf{Alg}(#1)}
\newcommand{\Coalg}[1]{\mathsf{Coalg}(#1)}

\begin{document}

\title[Many-valued coalgebraic logic over semi-primal varieties]{Many-valued coalgebraic logic over semi-primal varieties}
\author[A.~Kurz]{Alexander Kurz\lmcsorcid{0000-0002-8685-5207}}[a]
\author[W.~Poiger]{Wolfgang Poiger\lmcsorcid{0009-0002-8485-5905}}[b]
\author[B.~Teheux]{Bruno Teheux\lmcsorcid{0000-0002-3007-3089}}[b]
\address{Chapman University, 1 University Drive, 92866 Orange, California, USA}
\email{akurz@chapman.edu}  
\address{University of Luxembourg, 6 Avenue de la Fonte, L-4364 Esch-sur-Alzette, Luxembourg}	
\email{wolfgang.poiger@uni.lu, bruno.teheux@uni.lu}  

\begin{abstract}
  \noindent We study many-valued coalgebraic logics with semi-primal algebras of truth-degrees. We provide a systematic way to lift endofunctors defined on the variety of Boolean algebras to endofunctors on the variety generated by a semi-primal algebra. We show that this can be extended to a technique to lift classical coalgebraic logics to many-valued ones, and that (one-step) completeness and expressivity are preserved under this lifting. For specific classes of endofunctors, we also describe how to obtain an axiomatization of the lifted many-valued logic directly from an axiomatization of the original classical one. In particular, we apply all of these techniques to classical modal logic.          
\end{abstract}
\maketitle
\section*{Introduction}\label{sec:Introduction}
Research on \emph{many-valued modal logics} has been very active in recent years (see \cite{Fitting1991, DiaconescuGeorgescu2007, Priest2008, Maruyama2009, CaicedoRodriguez2010, Bou2011, HansoulTeheux2013, RivieccoJung2017, EstevaGodoVidal2017, MartiMetcalfe2018}, to name a few). In this paper, we study these logics from the perspective of \emph{coalgebraic logic}, the generalization of modal logic introduced by Moss in 1999 \cite{Moss1999}.

There are three distinct approaches to coalgebraic logic. The \emph{relation lifting approach} was introduced by Moss himself, and the \emph{predicate lifting} approach was initiated by Pattinson in \cite{Pattinson2003}. A unifying framework for both of these is found in the \emph{abstract approach} \cite{KupkeKurzPattinson2004,BonsangueKurz2005,KurzRosicky2012}. For classical modal logic, all three approaches have been fruitfully followed, resulting in interesting insights, generalizations and novel proof techniques (for a general overview of coalgebraic logic and a large collection of literature we refer the reader to \cite{KupkePattinson2011} and the bibliography therein). Thus, it is all the more surprising that very little research on many-valued coalgebraic logic exists thus far. Examples are  \cite{BilkovaKurzPetrisan2013} following the relation lifting approach and \cite{BilkovaDostal2016,LinLiau2022} following the predicate lifting approach. To the best of the authors knowledge, this paper  (together with the closely related \cite{KurzPoiger2023} by the first two authors), takes the first steps towards many-valued coalgebraic logic following the \emph{abstract approach}. 

In the classical setting, an \emph{abstract coalgebraic logic} (Definition~\ref{defin:AbstractCoalgebraicLogic}) for an endofunctor $\func{T}\colon \Set \to \Set$ is a pair $(\func{L}, \delta)$ consisting of an endofunctor $\func{L}\colon \BA \to \BA$ (essentially determining syntax) together with a natural transformation $\delta$ which, over the usual dual adjunction between $\Set$ and $\BA$, is used to relate $\func{T}$-coalgebras to $\func{L}$-algebras (essentially determining semantics). Important properties of coalgebraic logics, like \emph{one-step completeness} \cite{Pattinson2003,KupkeKurzPattinson2004} (Definition~\ref{defin:oneStepCompleteness}) and \emph{expressivity} \cite{Pattinson2004,Klin2007,Schroeder2008,JacobsSokolova2009} (Definition~\ref{defin:expressivity}) directly correspond to properties of $\delta$.

To retrieve classical modal logic as an example of an abstract coalgebraic logic, let $\func{T} = \mathcal{P}$ be the covariant powerset functor (whose coalgebras are Kripke frames), let $\func{L}$ be the functor which sends a Boolean algebra $\alg{B}$ to the free Boolean algebra generated by the set of formal expressions $\{ \Box b \mid b\in B \}$ modulo the familiar equations $\Box (b_1 \wedge b_2) \approx \Box b_1 \wedge \Box b_2$ and $\Box 1 \approx 1$ (whose algebras are modal algebras), and let the natural transformation $\delta$ take a Kripke frame to its complex modal algebra (Example~\ref{exam:classicalModalLogic}). This is not only an example of an abstract coalgebraic logic, but even of a \emph{concrete coalgebraic logic} (Definition~\ref{defin:ConcreteCoalgebraicLogic}), where $\func{L}$ is defined in terms of a \emph{presentation by operations and equations} \cite{BonsangueKurz2006,KurzPetrisan2010, KurzRosicky2012}, which essentially corresponds to an axiomatization of the corresponding variety of algebras. In this paper, to move from the classical to the many-valued setting, we replace the variety $\BA$ generated by the two-element algebra $\alg{2}$ by varieties $\var{A} = \variety{\alg{D}}$ generated by other finite algebras $\alg{D}$ of truth-degrees.  

More specifically, we discuss coalgebraic logics $(\func{L'},\delta')$ for which $\func{L}'\colon \var{A} \to \var{A}$ is an endofunctor on the variety generated by a \emph{semi-primal algebra} (Definition~\ref{def: semi-primal-algebra}) with an underlying bounded lattice. Semi-primal algebras were introduced by Foster and Pixley \cite{FosterPixley1964a} to generalize \emph{primal algebras}, which themselves were introduced by Foster \cite{Foster1953a} as immediate generalizations of the two-element Boolean algebra. Although usually not explicitly mentioned, modal extensions of semi-primal algebras have been studied before in a number of papers on finitely-valued modal logic. Let us mention some of them to put our work into context. Standard examples of semi-primal algebras in many-valued logic are the finite MV-chains, whose modal extensions have been studied in \cite{HansoulTeheux2013}, or the \L ukasiewicz-Moisil chains, whose extensions by tense operators have been studied in \cite{DiaconescuGeorgescu2007}. Furthermore, any extension of a $\FLew$-algebra by certain unary operations is semi-primal (Theorem~\ref{thm:semiprimal_characterizations}). For the semi-primal algebras thus obtained from finite Heyting algebras, modal extensions have been studied in \cite{Maruyama2009}. Modal extensions of $\FLew$-algebras in all generality have been studied in \cite{Bou2011}. The recent paper \cite{LinLiau2022} on coalgebraic many-valued logic also assumes the algebra of truth-degrees to be a $\FLew$-algebra either extended by these unary operations (rendering it semi-primal) or, even stronger, extended by the Baaz-delta \cite{Baaz1996} and all constants (even rendering it primal). It is, therefore, strongly related to our work.      

For primal algebras $\alg{D}$, in \cite{KurzPoiger2023} the first two authors made use of \emph{Hu's Theorem}, stating that the variety generated by a primal algebra is categorically equivalent to $\BA$, to lift coalgebraic logics over $\BA$ to ones over $\var{A}$. Here, we generalize the results therein to semi-primal algebras, where things get more involved in the absence of Hu's Theorem. 
Luckily, if $\var{A}$ is generated by a semi-primal algebra, there still is a useful Stone-type dual equivalence between $\var{A}$ and a category of structured Stone spaces $\StoneD$ (Definition~\ref{defin:StoneD}). In \cite{KurzPoigerTeheux2023} we explored the category-theoretical relationship between this duality and Stone's original duality, laying the ground work for this paper. Indeed, using the \emph{subalgebra adjunctions} \cite[Subsection 4.4]{KurzPoigerTeheux2023} we show that it is possible to lift functors $\func{T}\colon \Stone \to \Stone$ and $\func{L} \colon \BA \to \BA$ to functors $\func{T'}\colon \StoneD \to \StoneD$ and $\func{L}'\colon \var{A} \to \var{A}$, respectively. This allows us to systematically lift algebra-coalgebra dualities building on Stone duality, for example from \cite{KupkeKurzVenema2003,BezhanishviliDeGroot2022}, to the semi-primal level. In particular, the lifted `semi-primal version' of Jónsson-Tarski duality thus obtained coincides with the one established directly in \cite{Maruyama2012}. 

We also use the subalgebra adjunctions to lift classical coalgebraic logics $(\func{L},\delta)$ to many-valued coalgebraic logics $(\func{L}',\delta')$ over $\var{A}$ in a systematic way and show that one-step completeness and expressivity of $(\func{L}',\delta')$ follows directly from the corresponding properties of $(\func{L},\delta)$. Furthermore, we show that $\func{L}'$ has a presentation by operations and equations given that $\func{L}$ has one, and we demonstrate how, in certain cases, one can directly obtain such a presentation from the original one. For example, this allows us to generalize algebraic completeness results for finitely-valued modal logics (over crisp Kripke frames) as in \cite{Maruyama2009, HansoulTeheux2013}. We emphasize again that, in the theory developed here, such many-valued completeness (and expressivity) results are all \emph{direct consequences} of the corresponding completeness (expressivity) result for classical modal logic (Example~\ref{subsubs:classicalModalLogic}). We believe that results like these demonstrate that it is worthwhile to study many-valued logics coalgebraically, and we hope that this work inspires more future research in similar directions.

To put our work into its larger context, we believe that this research should be of interest to the community for a range of potential applications, from artificial intelligence and cyber-physical systems to the reasoning about software quality. The coalgebraic generalization also broadens the range of potential applications of many-valued modal logic, for example in modelling fuzzy preferences \cite{Vidal2020}, coalitional power \cite{KroupaTeheux2017} or searching games with errors \cite{Teheux2014}. Other applications of many-valued reasoning involve, for example, semiring-based algorithms for solving soft constraints (see, \emph{e.g.}, \cite{SchiendorferKAR18} for a recent example). From the point of view of some of the above-mentioned applications, a restriction of our approach is that the algebra of truth-degrees is finite and, correspondingly, the topological duality is zero-dimensional. Future work will be dedicated to investigating whether the techniques we develop to extend Boolean modal logics to many-valued modal logics can be further generalized to a continuum of truth-degrees, opening up even more possible applications, for example via its connections to metric behavioural theories (see, \emph{e.g.}, \cite{Baldan2018}).
 
The paper is structured as follows. In Section~\ref{sec:Preliminaries}, we give overviews of coalgebraic logic (Subsection~\ref{subsec: CoalgebraicModalLogic}), semi-primal algebras (Subsection~\ref{subsec: Semi-primal algebras}) and the Stone-type dualities for varieties they generate (Subsection~\ref{subs:SemiPrimalDuality}). In Section~\ref{sec:LiftingDualities}, we explain how to lift functors $\func{T}\colon \Stone \to \Stone$ and $\func{L}\colon \BA \to \BA$ to ones $\func{T}'\colon \StoneD \to \StoneD$ and $\func{L}'\colon \var{A} \to \var{A}$ (Definition~\ref{def:canonicalLiftingtop}, Proposition~\ref{prop:LiftingStoneMonoEpi}), and how to lift algebra-coalgebra dualites building on Stone duality to the semi-primal level (Theorem~\ref{thm:LiftingDualities}). In Section 4, we show how to lift abstract coalgebraic logics $(\func{L},\delta)$ based on $\Set$ and $\BA$ to abstract coalgebraic logics $(\func{L}',\delta')$ based on $\SetD$ and $\var{A}$ (Definition~\ref{def:LiftingAbstractCoalgLogic}). In particular, we show that one-step completeness and expressivity are preserved under this lifting (Theorems~\ref{thm:OneStepComplPreservation} and \ref{thm:ExpressPreservation}). In Section~\ref{sec:MVConcrete}, we show how to obtain an axiomatization for such a lifted logic from an axiomatization of the original one in some special cases (Theorems~\ref{thm:liftPresentationTaus} and \ref{thm:liftPresentationEtas}), together with some specific examples including classical modal logic (Subsection~\ref{subsec:applications}). Lastly, In Section~\ref{sec:conclusion}, we conclude the paper with some open questions and suggestions for further research.           
  
\section{Preliminaries}\label{sec:Preliminaries}
In Subsection~\ref{subsec: CoalgebraicModalLogic}, we give an overview of algebraic semantics for coalgebraic logics \cite{KupkeKurzPattinson2004}, one-step completeness  and expressivity. As an example, we recall how classical modal logic can be framed in this context. In Subsection~\ref{subsec: Semi-primal algebras}, we recall the definition of semi-primal algebras \cite{FosterPixley1964a} and give some examples of semi-primal algebras related to logic. In Subsection~\ref{subs:SemiPrimalDuality}, we describe the Stone-type topological duality \cite{KeimelWerner1974, ClarkDavey1998} for semi-primal varieties. For more information on semi-primal varieties we refer the reader to our previous paper \cite{KurzPoigerTeheux2023}, in which much of the `groundwork' for the present paper has been laid.   
     
\subsection{Coalgebraic modal logic}\label{subsec: CoalgebraicModalLogic}
Coalgebras offer a convenient category-theoretical framework to describe various transition-systems. A general theory of \emph{universal coalgebra} similar to that of universal algebra can be found in \cite{Rutten2000}. \emph{Coalgebraic (modal) logic} was introduced by Moss \cite{Moss1999} in 1999 and has quickly developed into an active research area (for an overview see, \emph{e.g.}, \cite{KupkePattinson2011}). In this subsection, we recall the basic definitions of algebras and coalgebras for a functor, as well as coalgebraic logics. Note that, as mentioned in the introduction, our approach to coalgebraic modal logic is the `abstract' one introduced in \cite{KupkeKurzPattinson2004} and further developed (among others) in \cite{BonsangueKurz2005,KurzRosicky2012}.

\begin{defi}\label{def:T-coalgebra}
Given a category $\cate{C}$ and an endofunctor $\func{T}\colon \cate{C} \to \cate{C}$, a $\func{T}$-\emph{coalgebra} is a $\cate{C}$-morphism $\gamma\colon X \to \func{T}(X)$, where $X\in \cate{C}$. Given another $\func{T}$-coalgebra $\gamma'\colon X'\to \func{T}(X')$, a $\func{T}$-coalgebra morphism $\gamma \to \gamma'$ is a $\cate{C}$-morphism $f\colon X \to X'$ for which the square
$$ 
\begin{tikzcd}
X \arrow[r, "\gamma"] \arrow[d, "f"']
& \func{T}(X) \arrow[d, "\func{T}f"] \\
X' \arrow[r,"\gamma'"]
&  \func{T}(X')
\end{tikzcd}  
$$
commutes. We denote by $\Coalg{\func{T}}$ the category of $\func{T}$-coalgebras with these morphisms.
\end{defi}

The following well-known example relates this to classical modal logic. 

\begin{exa}\label{exam:KripkeFrames}
A \emph{Kripke frame} is a pair $(X,R)$ where $X$ is a set and $R$ is a binary relation on $X$. A \emph{bounded morphism} (also called \emph{p-morphism}) between Kripke frames $(X_1,R_1)$ and $ (X_2,R_2)$ is a relation-preserving map $f\colon X_1\to X_2$ which, in addition, satisfies
$$
f(x)R_2y \Rightarrow \text{ there exists } x'\in X_1 \text{ such that } xRx' \text{ and } f(x') = y.
$$
It is well-known that the category $\Krip$ of Kripke frames with bounded morphisms is isomorphic to $\Coalg{\mathcal{P}}$ where $\mathcal{P}\colon \Set \to \Set$ is the covariant powerset functor.

A closely related example, the \emph{descriptive general frames}, are coalgebras for the \emph{Vietoris functor} $\mathcal{V}\colon \Stone \to \Stone$ on the category of Stone spaces (see \cite{KupkeKurzVenema2003}).
\end{exa}

Dealing with non-normal logics, Kripke semantics are usually replaced by the more general neighborhood semantics (see, \emph{e.g.}, \cite{Pacuit2017}).

\begin{exa}\label{exam:NeighborhoodFrames}
A \emph{neighborhood frame} is a pair $(X, N)$ where $X$ is a set and $N\colon X \to \mathcal{P}\mathcal{P}(X)$ sends $x$ to the collection of its neighborhoods $N(x)$. The category of neighborhood frames is isomorphic to the category of coalgebras for the \emph{neighborhood functor} $\mathcal{N} = \wp\circ \wp$, where $\wp\colon \Set \to \Set$ is the contravariant powerset functor (see, \emph{e.g.}, \cite{HansenKupkePacuit2009}).  
\end{exa}       

Coalgebras usually provide the structures on which we interpret formulas semantically. Algebras, on the other hand, tend to be closer to the syntactical side.
 
\begin{defi}\label{def:L-algebra}
Given a category $\cate{C}$ and an endofunctor $\func{L}\colon \cate{C} \to \cate{C}$, an $\func{L}$-\emph{algebra} is a $\cate{C}$-morphism $\alpha\colon\func{L}(A)\to A$ for some $A\in \cate{C}$. Given another $\func{L}$-algebra $\alpha'\colon \func{L}(A')\to A'$, an $\func{L}$-algebra morphism $\alpha \to \alpha'$ is a $\cate{C}$-morphism $h\colon A \to A'$ for which the square
$$ 
\begin{tikzcd}
\func{L}(A) \arrow[r, "\alpha"] \arrow[d, "\func{L}h"']
& A \arrow[d, "h"] \\
\func{L}(A') \arrow[r,"\alpha'"]
&  A'
\end{tikzcd}  
$$
commutes. We denote by $\Alg{\func{L}}$ the category of $\func{L}$-algebras with these morphisms.
\end{defi}

The usual `algebraic counterpart' of Example~\ref{exam:KripkeFrames} is given by the following.
\begin{exa}\label{exam:ModalAlgebras}
A \emph{modal algebra} is a pair $(\alg{B}, \Box)$, where $\alg{B}$ is a Boolean algebra and $\Box \colon B \to B$ preserves the top-element $1$ and finite meets. It was shown in \cite[Proposition 3.17]{KupkeKurzVenema2003} that the category of modal algebras with $\Box$-preserving homomorphisms is equivalent to the category $\Alg{\func{L}}$ for the functor $\func{L}\colon \BA\to \BA$ which assigns to a Boolean algebra $\alg{B}$ the free Boolean algebra generated by the underlying meet-semilattice of $\alg{B}$. 

Equivalently, the functor $\func{L}$ can be described in terms of a \emph{presentation by operations and equations} \cite{BonsangueKurz2006} as follows. For a Boolean algebra $\alg{B}$, the Boolean algebra $\func{L}(\alg{B})$ is the free Boolean algebra generated by the set of formal expressions $\{ \Box b \mid b\in \alg{B} \}$, modulo the equations $\Box 1 \approx 1$ and $\Box(b_1 \wedge b_2) \approx \Box b_1 \wedge \Box b_2$.     
\end{exa}

Similarly, we can describe the `algebraic counterpart' of Example~\ref{exam:NeighborhoodFrames} as follows. 

\begin{exa}\label{exam:NeighborAlgebras}
A \emph{neighborhood algebra} is a pair $(\alg{B}, \bigtriangleup)$, where $\alg{B}$ is a Boolean algebra and $\bigtriangleup\colon B \to B$ is an arbitrary operation. The category of neighborhood algebras is equivalent to the category of algebras for the functor $\func{L}\colon \BA \to \BA$ which has a presentation by one unary operation and no (that is, the empty set of) equations.
\end{exa}

We now recall from \cite{KupkeKurzPattinson2004} how categories of algebras and coalgebras are related in the context of coalgebraic logic. Similar to \cite{BonsangueKurz2006}, we define \emph{abstract} coalgebraic logics as follows.  

\begin{defi}\label{defin:AbstractCoalgebraicLogic}
Let $\cate{C}$ be a concrete category and let $\cate{V}$ be a variety of algebras. Let the functors $\func{P}\colon \cate{C}^\mathrm{op} \to \cate{V}$ and $\func{S}\colon \cate{V} \to \cate{C}^\mathrm{op}$ form a dual adjunction $\func{S}\dashv \func{P}$ (we will always identify them with contravariant functors between $\cate{C}$ and $\cate{V}$). Let $\func{T}\colon \cate{C} \to \cate{C}$ be an endofunctor. 
An \emph{abstract coalgebraic logic for $\func{T}$} is a pair $(\func{L}, \delta)$ consisting of an endofunctor $\func{L}\colon \cate{V}\to \cate{V}$ and a natural transformation $\delta\colon \func{L}\func{P} \Rightarrow \func{P}\func{T}$. 
$$
\begin{tikzpicture}
  \node (Set) at (-2,0) {$\cate{C}$};
  \node (V) at (2,0) {$\var{V}$};
  \draw [->,out=15,in=165,looseness=0.3] (Set) to node[above]{$\func{P}$}  (V);
  \draw [->,out=195,in=345,looseness=0.3] (V) to node[below]{$\func{S}$}  (Set);
  \draw [->,out=225,in=135,looseness=5] (Set) to node[pos = 0.35,above]{$\func{T}\phantom{.....}$}  (Set);
  \draw [->,out=-45,in=45,looseness=5] (V) to node[pos = 0.35,above]{$\phantom{.....}\func{L}$}  (V);
\end{tikzpicture}
$$ 
If $\var{C} = \Set$, $\var{V} = \BA$ and $\func{P}$ and $\func{S}$ are as in Example~\ref{exam:classicalModalLogic} below, we call the abstract coalgebraic logic \emph{classical}.      
\end{defi} 

In Examples~\ref{exam:ModalAlgebras} and \ref{exam:NeighborAlgebras} we saw that it is convenient to have a \emph{presentation} of an endofunctor by operations and equations in the sense of \cite[Definition 6]{BonsangueKurz2006} (see also \cite{KurzPetrisan2010, KurzRosicky2012}). In this case, as in \cite{KurzRosicky2012}, we will talk about \emph{concrete} coalgebraic logics.

\begin{defi}\label{defin:ConcreteCoalgebraicLogic}
A \emph{concrete coalgebraic logic} is an abstract coalgebraic logic $(\func{L},\delta)$ together with a presentation of $\func{L}$ by operations and equations.
\end{defi}

For a useful categorical characterization of concrete coalgebraic logics, we recall \cite[Theorem 4.7]{KurzRosicky2012}, stating that an endofunctor on a variety has a presentation by operations and equations if and only if it preserves \emph{sifted colimits} (for an introduction to sifted colimits and their role in universal algebra we refer the reader to \cite{AdamekRosickyVitale2010}).

In an abstract coalgebraic logic $(\func{L},\delta)$, the natural transformation $\delta$ is used to relate $\func{T}$-coalgebras to $\func{L}$-algebras as follows. Starting with a $\func{T}$-coalgebra $\gamma\colon X\to \func{T}(X)$, we can first apply the (contravariant) functor $\func{P}$ to obtain $\func{P}\gamma \colon \func{P}\func{T}(X) \to \func{P}(X)$. Now the component $\delta_X\colon \func{LP}(X)\to \func{PT}(X)$ of $\delta$ is precisely what is needed to obtain an $\func{L}$-algebra
$$
\func{P}\gamma \circ \delta_X \colon \func{LP}(X) \to \func{P}(X).
$$      
In the following we recall how to retrieve classical modal logic in this form. 
\begin{exa}\label{exam:classicalModalLogic}
In the context of Definition~\ref{defin:AbstractCoalgebraicLogic}, let $\cate{C} = \Set$ be the category of sets and $\cate{V} = \BA$ be the category of Boolean algebras. The dual adjunction is given by the functors $\func{S}\colon \BA \to \Set$ and $\func{P} \colon \Set \to \BA$, where $\func{S}$ takes a Boolean algebra to its set of ultrafilters and $\func{P}$ takes a set to its powerset-algebra.

Let $\func{T} = \var{P}$ be the covariant powerset functor (whose coalgebras, by Example~\ref{exam:KripkeFrames}, correspond to Kripke frames). We consider the following concrete coalgebraic logic $(\func{L}, \delta)$ for $\var{P}$. The functor $\func{L}\colon \BA \to \BA$ is defined as in Example~\ref{exam:ModalAlgebras}, corresponding to modal algebras. For a set $X\in \Set$, the component $\delta_X \colon \func{L}\func{P}(X) \to \func{P}\mathcal{P} (X)$ of $\delta$ is given by 
$$ \Box Y \mapsto \{ Z \subseteq X \mid Z \subseteq Y \} \text{ for } Y\subseteq X.$$ 
Let the coalgebra $\gamma_R\colon X \to \var{P} (X)$ be identified with its corresponding Kripke frame given by $(X,R)$, where $\gamma_R(x) = \{ x'\in X \mid xRx' \}$. Untangling the definitions shows that the operator corresponding to the modal algebra 
$$
\func{P}\gamma \circ \delta_X \colon \func{L}\func{P} (X) \to \func{P} (X)
$$
is defined on the subset $Y\subseteq X$ by
$$
\Box Y =  \{ x\in X \mid \gamma_R(x) \subseteq Y \} = \{ x\in X \mid xRx' \Rightarrow x'\in Y \}.
$$
This is commonly known as the \emph{complex algebra} of the frame $(X,R)$ (see, \emph{e.g.}, \cite{BlackburnRijkeVenema2001}).        
\end{exa}

The analogous coalgebraic logic corresponding to neighborhood semantics is discussed in the following (for more information we refer the reader to \cite{BezhanishviliDeGroot2022}). 

\begin{exa}\label{exam:classicalNeighborhoodLogic}
Let $\func{S}\colon \BA \to \Set$ and $\func{P} \colon \Set \to \BA$ be as in the previous example. We consider the following abstract coalgebraic logic $(\func{L}, \delta)$ for $\var{N}$, the neighborhood functor from Example~\ref{exam:NeighborhoodFrames}. The functor $\func{L}\colon \BA \to \BA$ is defined as in Example~\ref{exam:NeighborAlgebras}, corresponding to neighborhood algebras. For a set $X\in \Set$, the component $\delta_X \colon \func{L}\func{P} (X) \to \func{P}\mathcal{N}(X)$ of $\delta$ is given by 
$$ 
\Box Y \mapsto \{ N \subseteq \mathcal{P}(X) \mid Y \in N \} \text{ for } Y\subseteq X.
$$ 
This corresponds to the way neighborhood algebras are obtained from neighborhood frames in \cite{Dosen1989}.   
\end{exa}      

Next we discuss some key-properties a coalgebraic logic may have, namely (one-step) completeness and expressivity.   

\begin{defi}\label{defin:oneStepCompleteness}\cite{Pattinson2003,KupkeKurzPattinson2004}
An abstract coalgebraic logic $(\func{L}, \delta)$ is called \emph{one-step complete} if $\delta$ is a monomorphism. 
\end{defi} 
If the coalgebraic logic $(\func{L},\delta)$ is concrete, the category of algebras $\Alg{\func{L}}$ forms a variety, whose equational logic can be equivalently described as a modal logic. Semantically, we interpret such modal formulas as follows. Given a coalgebra $\gamma\colon X \to \func{T}(X)$, we use $\delta$ to associate the corresponding  $\func{L}$-algebra to it. The initial $\func{L}$-algebra $\func{L}(I) \to I$ exists and the unique morphism
$$ 
\begin{tikzcd}
\func{L}(I) \arrow[r] \arrow[d, "{\func{L}[[\cdot]]}"']
& I \arrow[d, "{[[\cdot]]}"] \\
\func{LP}(X) \arrow[r]
&  \func{P}(X)
\end{tikzcd}  
$$
yields the \emph{interpretation} of formulas $[[\cdot]]\colon I\to \func{P}(X)$. Via the dual adjunction, this corresponds to $\theory\colon X\to \func{S}(I)$, which assigns a point of $X$ to its \emph{theory}. In \cite[Theorem 6.15]{KurzPetrisan2010} it is shown that, for any classical concrete coalgebraic logic, one-step completeness implies completeness for $\func{T}$-coalgebras in this sense (see also \cite{Pattinson2003} for the point of view of predicate liftings). For example, the concrete classical coalgebraic logics described in Examples~\ref{exam:classicalModalLogic} and \ref{exam:classicalNeighborhoodLogic} are one-step complete and, therefore, complete.  

Similar to how completeness of $(\func{L}, \delta)$ is related to $\delta$ being a monomorphism, expressivity \cite{Pattinson2004,Klin2007,Schroeder2008,JacobsSokolova2009} can be related to the \emph{adjoint transpose} of $\delta$ being a monomorphism. We use \cite[Theorem 4]{JacobsSokolova2009} as definition of expressivity. Recall that, in the setting of Definition~\ref{defin:AbstractCoalgebraicLogic}, the adjoint transpose $\delta^\dagger\colon \func {TS}\Rightarrow \func{SL}$ is obtained from $\delta$ as composition 
$$
\begin{tikzcd}
\func{TS} \arrow[r, "\varepsilon\func{TS}"] & \func{SPTS} \arrow[r, "\func{S}\delta\func{S}"] & \func{SLPS} \arrow[r, "\func{SL}\eta"] & \func{SL} 
\end{tikzcd}
$$
where $\varepsilon$ and $\eta$ are the unit and counit of the adjunction.
\begin{defi}\label{defin:expressivity}
Let $(\func{L},\delta)$ be an abstract coalgebraic logic for $\func{T}\colon \cate{C} \to \cate{C}$ satisfying the following conditions. 
\begin{enumerate}
\item The category $\Alg{\func{L}}$ has an initial object. 
\item The category $\cate{C}$ has $(\var{M},\var{E})$-factorizations with $\var{M}$ a collection of monomorphisms and $\var{E}$ a collection of epimorphisms.   
\item The functor $\func{T}$ preserves members of $\var{M}$.  
\end{enumerate} 
We say that $(\func{L},\delta)$ is \emph{expressive} if every component of the adjoint transpose $\delta^\dag$ of $\delta$ is in $\var{M}$.
\end{defi} 

Assuming the coalgebraic logic is concrete, expressivity is also known as the \emph{Hennessy-Milner property}, stating that for every $\func{T}$-coalgebra $\gamma \colon X \to \func{T}(X)$, two states $x,y\in X$ have the same theory if and only if they are behaviourally equivalent (bisimilar). Here, we call $x_1\in X_1$ and $x_2\in X_2$ \emph{behaviourally equivalent} if there exist coalgebra morphisms $f_1\colon X_1 \to Y$ and $f_2 \colon X_2 \to Y$ (into the same coalgebra) with $f_1(x_1) = f_2(x_2)$.    

The logic of Example~\ref{exam:classicalModalLogic} is not expressive, but becomes expressive if restricted to \emph{image-finite} Kripke frames (that is, the powerset functor is replaced by the \emph{finite powerset functor} $\mathcal{P}_\omega$). Similar results for the logic defined in Example~\ref{exam:classicalNeighborhoodLogic} and the appropriate definition of image-finite neighborhood frames can be found in \cite{HansenKupkePacuit2009}.

Classical coalgebraic logics are partially described by endofunctors on the variety $\BA$, which is generated as a variety by the two-element Boolean algebra $\alg{2}$ of truth-degrees, \emph{i.e.}, $\BA = \variety{\alg{2}}$. To deal with many-valued logics we want to replace $\alg{2}$ by another (in our case, finite) algebra of truth-degrees $\alg{D}$. We will explore the scenario where $\alg{D}$ is \emph{semi-primal}, which implies that $\variety{\alg{D}}$ behaves similarly to (but is usually not categorically equivalent to) $\BA$. We give an overview of semi-primal algebras in the following subsection.  
\subsection{Semi-primal algebras}\label{subsec: Semi-primal algebras}
The study of primality and its variants is a classical topic in universal algebra (see, \emph{e.g.}, \cite{Quackenbush1979} and \cite[Chapter IV]{BurrisSankappanavar1981}), which evolved from Foster's generalized `Boolean' theory of universal algebras \cite{Foster1953a, Foster1953b}, where primal algebras were introduced. A finite algebra $\alg{P}$ is \emph{primal} if every operation $P^n \to P$ (with $n \geq 1$) on the carrier set of $\alg{P}$ is term-definable in $\alg{P}$. As suggested by Foster's original title, arguably the most important example of a primal algebra is the two-element Boolean algebra $\alg{2}$. Indeed, Hu's theorem \cite{Hu1969, Hu1971} states that a variety of algebras $\var{A}$ is categorically equivalent to the variety $\BA$ of Boolean algebras if and only if it is generated by some primal algebra $\alg{P} \in \var{A}$ as a variety, \emph{i.e.}, $\var{A} = \variety{\alg{P}}$.

Foster and Pixley initiated the study of weakened forms of primality, introducing \emph{semi-primal} algebras in \cite{FosterPixley1964a}. Unlike primal algebras, semi-primal algebras $\alg{D}$ can have proper subalgebras. Note that every term-definable operation $f\colon D^n \to D$ necessarily \emph{preserves subalgebras}, meaning that if $\alg{S} \leq \alg{D}$ is a subalgebra, then $f(S^n) \subseteq S$ holds. Of course, only such operations can possibly be term-definable in $\alg{D}$. 
\begin{defi}\label{def: semi-primal-algebra}
A finite algebra $\alg{D}$ is \emph{semi-primal} if for every $n\geq 1$, every operation $f\colon D^n\to D$ which preserves subalgebras is term-definable in $\alg{D}$.
\end{defi} 
Next we recall some equivalent characterizations of semi-primality. The \emph{ternary discriminator} $t\colon L^3 \to L$ is given by
$$
t(x,y,z) = \begin{cases}
z & \text{ if } x = y \\
x & \text{ if } x \neq y.
\end{cases}
$$ 
An algebra in which the ternary discriminator is term-definable is called \emph{discriminator algebra} and finite discriminator algebras are also called \emph{quasi-primal}.  

Furthermore, recall that an \emph{internal isomorphism} of an algebra $\alg{D}$ is an isomorphism $\varphi\colon \alg{S}_1 \to \alg{S}_2$ between any two (not necessarily distinct) subalgebras $\alg{S}_1$ and $\alg{S}_2$ of $\alg{D}$. For example, if $\alg{S} \leq \alg{D}$ is a subalgebra, then the identity $\id_S$ is an internal isomorphism of $\alg{D}$.   

Lastly, recall that a variety of algebras is called \emph{arithmetical} if every member of the variety has a distributive lattice of commuting congruences.

\begin{thm}\label{thm:semiprimal_characterizations}
Let $\alg{D}$ be a finite algebra. Then the following are equivalent. 
\begin{enumerate}
\item $\alg{D}$ is semi-primal. 
\item $\alg{D}$ is quasi-primal and the only internal isomorphisms of $\alg{D}$ are the identities on subalgebras of $\alg{D}$ \cite{Pixley1971}.
\item $\alg{D}$ is simple, the variety generated by $\alg{D}$ is arithmetical and the only internal isomorphisms of $\alg{D}$ are the identities on subalgebras on $\alg{D}$ \cite{FosterPixley1964b}.
\end{enumerate}
If $\alg{D}$ is based on a bounded lattice $\alg{D}^\flat = \langle L, \wedge, \vee, 0, 1\rangle$, the following are also equivalent to the above conditions \cite{Foster1967,KurzPoigerTeheux2023}. 
\begin{enumerate}
\item[(4)] For every $d \in \alg{D}$, the unary operation $T_d \colon D\to D$ defined by 
$$
T_d (x) =  \begin{cases}
1 & \text{ if } x = d \\
0 & \text{ if } x \neq d
\end{cases}
$$
is term-definable in $\alg{D}$.
\item[(5)] The unary term $T_0$ is term-definable and for every $d\in \alg{D}$ the unary operation $\tau_d \colon D\to D$ defined by 
$$
\tau_d (x) =  \begin{cases}
1 & \text{ if } x \geq d \\
0 & \text{ if } x \not\geq d
\end{cases}
$$
is term-definable in $\alg{D}$. 
\end{enumerate} 
\end{thm}

Some examples of semi-primal algebras related to many-valued logic are provided below (for more examples see, \emph{e.g.}, \cite{Burris1992} or \cite[Subsection 2.3]{KurzPoigerTeheux2023}). Our first example is the algebraic counterpart of \L ukasiewicz finitely-valued logic (for more details we refer the reader to \cite{Cignoli2000}).      

\begin{exa}\label{exam:LukasChains}
The \emph{$n$-th Łukasiewicz chain} is given by 
$$
\lucas_n = \langle\{0, \tfrac{1}{n}, \dots, \tfrac{n-1}{n},1\},\wedge, \vee, \oplus,\odot,\neg, 0, 1\rangle,
$$
where $x\oplus y = \mathrm{min}(x+y,1)$, $x\odot y = \mathrm{max}(x+y-1, 0)$ and $\neg x = 1 - x$. 
For all $n\geq 1$, the algebra $\lucas_n$ is semi-primal (see \cite[Proposition 2.1]{Niederkorn2001}).
The subalgebras of $\lucas_n$ are exactly of the form $\lucas_d$, where $d$ is a divisor of $n$.
\end{exa}

Many-valued modal logic using $\lucas_n$ as algebra of truth-degrees were studied from an algebraic perspective in \cite{HansoulTeheux2013}.
      
Intending to provide algebraic semantics for \L ukasiewicz finitely-valued logic, Moisil studied the algebras discussed in our next example. However, in general they turned out to encompass more than that (see \cite{Cignoli1982}), and the logic corresponding to these algebras is nowadays commonly named after Moisil.
\begin{exa}\label{exam:MoisilChains}
The \emph{$n$-th \L ukasiewicz-Moisil chain} is given by 
$$
\alg{M}_n = \langle\{0, \tfrac{1}{n}, \dots, \tfrac{n-1}{n},1\}, \wedge, \vee, \neg, 0, 1, (\tau_{\frac{i}{n}})_{i=1}^n\rangle,
$$
where $\neg x = 1-x$ and the unary operations $\tau_{\frac{i}{n}}$ are the ones from Theorem~\ref{thm:semiprimal_characterizations} (5), which is also used to easily check that, for all $n\geq 1$, the algebra $\alg{M}_n$ is semi-primal.  
\end{exa} 
For general information about these algebras see \cite{Boicescu1991}. Extensions of $\alg{M}_n$ by tense operators have been studied in \cite{DiaconescuGeorgescu2007}.

A common framework for various many-valued logics, including Example~\ref{exam:LukasChains}, is provided by the following class of algebras. 

\begin{defi}\label{defin:FLewAlgebras}
A \emph{$\FLew$-algebra} is an algebra $\alg{F} = \langle F, \wedge, \vee, 0, 1, \odot, \rightarrow\rangle$ such that  
\begin{itemize}
\item $\langle F,\wedge, \vee, 0, 1\rangle$ is a bounded lattice, 
\item $\langle F, \odot, 1\rangle$ is a commutative monoid,
\end{itemize}
and the \emph{residuation law}
$x\odot y \leq z \Leftrightarrow x \leq y\rightarrow z$
is satisfied in $\alg{F}$. 
\end{defi} 
In any $\FLew$-algebra, we use the common notation $\neg x := x \rightarrow 0$. The following, together with Theorem~\ref{thm:semiprimal_characterizations}(2), offers various ways to systematically identify semi-primal $\FLew$-algebras. 

\begin{prop}\label{prop:QuasiPrimalFLewAlgebras}
Let $\alg{F}$ be a finite $\FLew$-algebra. 
\begin{enumerate}[(i)]
\item $\alg{F}$ is quasi-primal if and only if there exists some $n\geq 1$ such that the equation $x \vee \neg (x^n) \approx 1$ is satisfied in $\alg{F}$ \cite[Theorem 3.10]{Kowalski2004}. 
\item $\alg{F}$ is quasi-primal if and only if $T_1$ is term-definable in $\alg{F}$ \cite{DaveySchumann1991}.
\item If no element $a\in \alg{F}{\setminus}\{0,1\}$ satisfies $a \odot a = a$, then $\alg{F}$ is quasi-primal. The converse also holds if $\alg{F}$ is based on a chain \cite[Corollary 2.16]{KurzPoigerTeheux2023}.  
\end{enumerate}   
\end{prop}

Modal extensions of Heyting algebras expanded with all $T_d$ as in Theorem~\ref{thm:semiprimal_characterizations}(4) were studied in \cite{Maruyama2009}. A general study of many-valued modal logics over finite $\FLew$-algebras was initiated in \cite{Bou2011}. The recent paper \cite{LinLiau2022} studies them from a coalgebraic (more specifically, predicate lifting) perspective. Since that paper assumes the language to either contain all unary terms $T_d$ or all constants and the unary term $T_1$ (commonly referred to as the \emph{Baaz-delta} due to \cite{Baaz1996}), the corresponding algebras are either semi-primal (by Theorem~\ref{thm:semiprimal_characterizations}) or primal (by the above proposition they are quasi-primal after adding $T_1$, and a quasi-primal algebra becomes primal if all constants are added to the signature).  

In the next subsection we discuss a Stone-type duality for varieties generated by semi-primal algebras. 
\subsection{Semi-primal topological duality}\label{subs:SemiPrimalDuality}  
A milestone result in algebraic logic is \emph{Stone duality}, the famous dual equivalence between the variety $\BA$ and the category $\Stone$ of Stone spaces (\emph{i.e.}, compact, Hausdorff and totally disconnected topological spaces) with continuous maps. Hu \cite{Hu1969,Hu1971} generalized this duality for any variety generated by a primal algebra. Using sheaf representations, Keimel and Werner \cite{KeimelWerner1974} described a similar duality for varieties generated by finitely many quasi-primal algebras. This duality can also be expressed in the language of \emph{natural dualities} \cite{ClarkDavey1998}, and it becomes particularly simple if the algebra is assumed to be not only quasi-primal but even semi-primal \cite[Theorem 3.3.13]{ClarkDavey1998}. Below, we present this duality in a self-contained manner, without referring to the theory of sheaves or natural dualities. We follow the presentation of our previous paper \cite{KurzPoigerTeheux2023}, which the reader may consult for further details.           

\begin{defi}\label{defin:StoneD}
Let $\alg{D}$ be a semi-primal algebra. The category $\StoneD$ has objects $(X,\val{v})$, where $X$ is a Stone space and $\val{v}\colon X \to \mathbb{S}(\alg{D})$ assigns a subalgebra of $\alg{D}$ to every point of $X$, such that for every $\alg{S} \in \mathbb{S}(\alg{D})$, the preimage $\val{v}^{-1}(\alg{S}{\downarrow})$ is closed. A morphism $f\colon (X_1,\val{v}_1) \to (X_2, \val{v}_2)$ is a continuous map $X_1 \to X_2$ which satisfies 
$$
\val{v}_2(f(x)) \leq \val{v}_1(x)
$$
for every $x\in X_1$. 
\end{defi}  

In the following, let $\alg{D}$ be a semi-primal algebra and let $\var{A} = \variety{\alg{D}}$ be the variety it generates. We define two contravariant functors $\Sigma'\colon \var{A} \to \StoneD$ and $\Pi' \colon \StoneD \to \var{A}$, which establish the topological duality.     

The functor $\Sigma' \colon \var{A} \to \StoneD$ is defined on objects $\alg{A}\in \var{A}$ by 
$$
\Sigma'(\alg{A}) = \big(\var{A}(\alg{A},\alg{D}), \val{im}\big),
$$
that is, the set of homomorphisms $u\colon \alg{A} \to \alg{D}$ and $\val{im}\colon \var{A}(\alg{A},\alg{D}) \to \mathbb{S}(\alg{D})$ taking $u$ to its image $u(D)$. The topology on $\var{A}(\alg{A},\alg{D})$ is generated by the subbasis
$$
[a:d] = \{ u \mid u(a) = d \},
$$
where $a$ and $d$ range over $A$ and $D$, respectively. To a homomorphism $h\colon \alg{A}_1 \to \alg{A}_2$ the functor $\Sigma'$ assigns the $\StoneD$-morphism 
$\Sigma' h \colon \var{A}(\alg{A}_2,\alg{D}) \to \var{A}(\alg{A}_1,\alg{D})$ given by 
$u \mapsto u \circ h. $ 

The functor $\Pi' \colon \StoneD \to \var{A}$ is defined on objects $(X,\val{v}) \in \StoneD$ by 
$$
\Pi'(X,\val{v}) = \{ \alpha \in \Stone(X, D) \mid \forall x\in X\colon \alpha(x) \in \val{v}(x) \},
$$
where $\Stone(X, D)$ denotes the set of continuous maps $X\to D$ (the latter equipped with the discrete topology) with componentwise operations (\emph{i.e.}, as a subalgebra of $\alg{D}^X$). To a morphism $f\colon (X_1,\val{v}_1) \to (X_2,\val{v}_2)$ the functor $\Pi'$ assigns the homomorphism $\alpha \mapsto \alpha \circ f$.

In particular, if $\alg{D} = \alg{2}$ is the two-element Boolean algebra, the above functors establish Stone duality. In this case, we will simply denote the functors by $\Sigma$ and $\Pi$ (instead of $\Sigma'$ and $\Pi'$).  

\begin{thm}\label{thm:semiPrimalDuality}\cite{KeimelWerner1974, ClarkDavey1998}
Let $\alg{D}$ be a semi-primal algebra without trivial subalgebras and $\var{A} = \variety{\alg{D}}$ the variety it generates. Then the functors $\Pi'$ and $\Sigma'$ establish a dual equivalence between $\StoneD$ and $\var{A}$.   
\end{thm}

Both \cite{KeimelWerner1974} and \cite{ClarkDavey1998} actually allow $\alg{D}$ to have trivial subalgebras. However, the description of the dual category gets slightly more involved, and it will not be necessary for us, since in the following we will always assume that $\alg{D}$ has an underlying bounded lattice structure (under this assumption, in \cite[Section 3]{KurzPoigerTeheux2023} we provide an alternative proof of the above theorem based on methods developed in \cite{Johnstone1982}). 

\section{Lifting Algebra-Coalgebra Dualities}\label{sec:LiftingDualities}

For the remainder of this paper, we always work under the following assumption, unless explicitly stated otherwise. 

\begin{asm}\label{ass:mainassumption}
The finite algebra $\alg{D}$ is semi-primal algebra and has a bounded lattice reduct $\langle D,\wedge, \vee, 0, 1 \rangle$. We use $\var{A} = \variety{\alg{D}}$ to denote the variety generated by $\alg{D}$.  
\end{asm}

In this section, we describe a canonical way to lift endofunctors on $\Stone$ to ones one $\StoneD$ and, dually, to lift endofunctors on $\BA$ to ones on $\var{A}$. In particular, if  $\func{T}\colon\Stone \to \Stone$ and $\func{L}\colon \BA \to \BA$ are dual (in the sense that $\func{L} \cong \Pi\func{T}\Sigma$), then their respective liftings $\func{T}'\colon \StoneD \to \StoneD$ and $\func{L}'\colon \var{A} \to \var{A}$ are dual as well (in the sense that $\func{L}' \cong \Pi'\func{T}'\Sigma'$). For example, the `semi-primal version' of Jónsson-Tarski duality which Maruyama established directly in \cite{Maruyama2012} can also be obtained by lifting the classical Jónsson-Tarski duality in the systematic way we describe here. Other dualities that can be lifted to the `semi-primal level' include Došen duality \cite{Dosen1989} as framed, among others, in \cite{BezhanishviliDeGroot2022} by algebras and coalgebras.  

In order to lift endofunctors, we will make use of the \emph{subalgebra adjunctions} described in \cite[Subsection 4.4]{KurzPoigerTeheux2023}. For an overview of the functors described in the following, see Figure~\ref{fig:1}. 
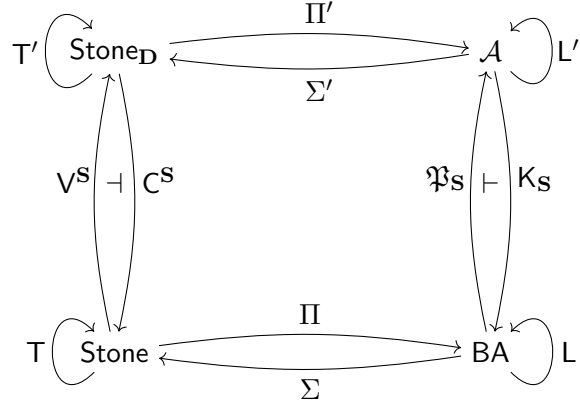
\begin{figure}[ht]
$$
\begin{tikzpicture}
  \node (Stone) at (-2.5,0) {$\Stone$};
  \node (BA) at (2.5,0) {$\BA$};
  \node (StoneD) at (-2.5,4) {$\StoneD$};
  \node (A) at (2.5,4) {$\var{A}$};
  \node (lad) at (-2.5,2.25) {$\dashv$};
  \node (rad) at (2.5,2.25) {$\vdash$};
  \draw [->, bend left = 8] (Stone) to node[above]{$\Pi$}  (BA);
  \draw [->, bend left = 8] (BA) to node[below]{$\Sigma$}  (Stone);
  \draw [->, bend left = 8] (StoneD) to node[above]{$\Pi'$}  (A);
  \draw [->, bend left = 8] (A) to node[below]{$\Sigma'$}  (StoneD);
  \draw [->, bend left = 12] (Stone) to node[above]{$\VS\phantom{.....}$}  (StoneD);
  \draw [->, bend left = 12] (StoneD) to node[above]{$\phantom{......}\CS$}  (Stone);
  \draw [->, bend left = 12] (BA) to node[above]{$\PS\phantom{......}$}  (A);
  \draw [->, bend left = 12] (A) to node[above]{$\phantom{......}\KS$}  (BA);
  \draw [->,out=225,in=135,looseness=5] (Stone) to node[pos = 0.35,above]{$\func{T}\phantom{.....}$}  (Stone);
  \draw [->,out=225,in=135,looseness=5] (StoneD) to node[pos = 0.35,above]{$\func{T}'\phantom{......}$}  (StoneD);
  \draw [->,out=-45,in=45,looseness=5] (BA) to node[pos = 0.35,above]{$\phantom{.....}\func{L}$}  (BA);
  \draw [->,out=-45,in=45,looseness=5] (A) to node[pos = 0.35,above]{$\phantom{.....}\func{L}'$}  (A);
\end{tikzpicture}
$$
\caption{Lifting algebra-coalgebra dualities via the subalgebra adjunctions.}
\label{fig:1}
\end{figure}
   
For every subalgebra $\alg{S} \in \mathbb{S}(\alg{D})$ the subalgebra adjunction corresponding to $\alg{S}$ is given by the following functors $\VS \dashv \CS$. The functor $\VS\colon \Stone\to\StoneD$ takes a Stone space $X$ to $(X,\val{v}^\alg{S})$ where $\val{v}^\alg{S}(x) = \alg{S}$ for all $x\in X$ (and $\VS$ is the identity on morphisms). Its right-adjoint $\CS\colon \StoneD \to \Stone$ takes $(X,\val{v}) \in \StoneD$ to the closed subspace $\CS(X,\val{v}) = \{ x\in X \mid \val{v}(x) \leq \alg{S} \}$ (and $\CS$ acts by restriction on morphisms). Note that $\U := \func{C}^\alg{D}$ is simply the \emph{forgetful functor} $(X,\val{v}) \mapsto X$.

The dual of $\VS$ is given by $\PS\colon \BA \to \var{A}$ which takes a Boolean algebra $\alg{B}$ to the \emph{Boolean power} $\alg{S}[\alg{B}]$ (see  \cite{Burris1975, KurzPoigerTeheux2023}). 

The dual of the forgetful functor $\U$ is given by the \emph{Boolean skeleton} $\Skel\colon \var{A} \to \BA$, taking an algebra $\alg{A}$ to the Boolean algebra defined on the subset 
$$ \Skel(A) = \{ a\in A \mid T_1(a) = a \} $$
together with the lattice operations inherited from $\alg{A}$ and negation $T_0$ (here, the unary terms $T_0$ and $T_1$ come from the ones in Theorem~\ref{thm:semiprimal_characterizations}). More generally, the dual of $\CS$ is given by $\KS\colon \var{A} \to \BA$ which takes $\alg{A}$ to the Boolean skeleton of an appropriate quotient of $\alg{A}$ (see \cite[Subsection 4.4]{KurzPoigerTeheux2023}). 

The subalgebra adjunctions can be used to reconstruct (up to isomorphism) an object $(X,\val{v})\in \StoneD$ from the information carried by all $\CS(X,\val{v})$ as a \emph{coend} (for the general theory of ends and coends see, \emph{e.g.}, \cite{MacLane1997}).

More specifically, considering $\mathbb{S}(\alg{D})$ ordered by inclusion, we define a coend diagram $\mathbb{S}(\alg{D})^\mathrm{op}\times \mathbb{S}(\alg{D}) \to \StoneD$ corresponding to $(X,\val{v})$ as follows. A pair of subalgebras $(\alg{S},\alg{T})$ gets assigned to $\V^{\alg{S}}\C^{\alg{T}}(X,\val{v})$ and if $\alg{S}_1 \leq \alg{S}_2$ and $\alg{T}_1 \leq \alg{T}_2$, then the inclusion $\C^{\alg{T}_1}(X,\val{v}) \hookrightarrow \C^{\alg{T}_2}(X,\val{v})$ yields a well-defined morphism $\V^{\alg{S}_2}\C^{\alg{T}_1}(X,\val{v}) \hookrightarrow \V^{\alg{S}_1}\C^{\alg{T}_2}(X,\val{v})$ in $\StoneD$.   

\begin{prop}\label{prop:CoendPresentation}
Let $(X,\val{v}) \in \StoneD$. Then $(X,\val{v})$ is isomorphic the coend of the diagram defined above, that is,
$$
(X,\val{v}) \cong  \int^{\alg{S}\in\mathbb{S}(\alg{D})} \VS \CS (X,\val{v}).
$$
\end{prop}         

\begin{proof}
The inclusion maps $\iota_\alg{S} \colon \CS(X,\val{v}) \hookrightarrow X$ are morphisms $\VS\CS(X,\val{v}) \hookrightarrow (X,\val{v})$ because if $x\in \CS(X)$ then $\val{v}(x) \leq \alg{S}$ and thus $\val{v}(\iota(x)) = \val{v}(x) \leq \alg{S} = \val{v}^\alg{S}(x)$. Since the diagram only contains inclusion maps, clearly this defines a cowedge. 

Now assume that $c_\alg{S}\colon \VS\CS(X,\val{v})\to(Y,\val{w})$ is another cowedge. Then the underlying  map of $c_\alg{D}\colon X \to Y$ yields a well-defined morphism $(X,\val{v}) \to (Y,\val{w})$. To see this, take $x\in X$ and note that $\val{w}(c_\alg{D}(x))\leq \alg{S}$ whenever $x\in \CS(X,\val{v})$, because the diagram 
$$
\begin{tikzcd}[row sep=4.5em,column sep=4.5em]
\V^{\alg{D}}\CS(X,\val{v}) \arrow[r, hookrightarrow] \arrow[d, hookrightarrow] & \VS\CS(X,\val{v}) \arrow[d, "c_\alg{S}"] \\
\V^\alg{D}\C^\alg{D}(X,\val{v}) \arrow[r, "c_\alg{D}"'] &  (Y,\val{w})
\end{tikzcd}
$$    
commutes. Therefore, we have $\val{w}(c_\alg{D}(x)) \leq \bigwedge\{\alg{S} \mid x\in \CS(X,\val{v})\} = \val{v}(x).$ 
\end{proof}

As an immediate consequence of this we get the following dual statement. 

\begin{cor}\label{cor:EndPresentation}
Let $\alg{A}\in \var{A}$. Then $\alg{A}$ is isomorphic to the end of the diagram dual to the coend diagram corresponding to $\Sigma'(\alg{A})$, that is, 
$$
\alg{A} \cong \int_{\alg{S}\in \mathbb{S}(\alg{D})} \PS\KS(\alg{A}).
$$  
\end{cor}

The presentations of Proposition~\ref{prop:CoendPresentation} and Corollary~\ref{cor:EndPresentation} yield canonical ways to lift functors from $\Stone$ to $\StoneD$ and from $\BA$ to $\var{A}$.

\begin{defi}\label{def:canonicalLiftingtop}
Let $\func{T}\colon \Stone \to \Stone$ and $\func{L}\colon \BA \to \BA$ be functors. 
\begin{enumerate}[(a)]
\item The \emph{lifting of $\func{T}$ to $\StoneD$} is the functor $\func{T}'\colon \StoneD \to \StoneD$ defined on objects by
$$
\func{T}'(X,\val{v}) = \int^{\alg{S}\in\mathbb{S}(\alg{D})} \VS \func{T} \CS (X,\val{v}).
$$ 
\item The \emph{lifting of $\func{L}$ to $\var{A}$} is the functor $\func{L}'\colon \var{A}\to \var{A}$ defined on objects by 
$$
\func{L}'(\alg{A}) = \int_{\alg{S}\in \mathbb{S}(\alg{D})} \PS \func{L}\KS(\alg{A}).
$$  
\end{enumerate}
The definitions of $\func{T}'$ and $\func{L}'$ on morphisms are discussed in the next paragraph.
\end{defi}

Let $f\colon (X_1,\val{v}_1)\to (X_2,\val{v}_2)$ be a $\StoneD$-morphism. Then we can define $\func{T}'f$ by the universal property of the coend, once we define a cowedge $\VS \func{T} \CS (X_1,\val{v}_1) \to \func{T}'(X_2,\val{v}_2)$ as follows. 
$$
\begin{tikzcd}[row sep=2em, column sep=3em]
\V^{\alg{S}_2}\func{T}\C^{\alg{S}_2}(X_1,\val{v}_1) \arrow[r] & \V^{\alg{S}_2}\func{T}\C^{\alg{S}_2}(X_2,\val{v}_2) \arrow[dr] & \\
\V^{\alg{S}_2}\func{T}\C^{\alg{S}_1}(X_1,\val{v}_1) \arrow[u] \arrow[d] \arrow[r] & \V^{\alg{S}_2}\func{T}\C^{\alg{S}_1}(X_2,\val{v}_2) \arrow[u] \arrow[d] & \func{T}'(X_2,\val{v}_2)\\
\V^{\alg{S}_1}\func{T}\C^{\alg{S}_1}(X_1,\val{v}_1) \arrow[r]& \V^{\alg{S}_1}\func{T}\C^{\alg{S}_1}(X_2,\val{v}_2) \arrow[ur] & \\
\end{tikzcd}
$$ 
Here we have $\alg{S}_1 \leq \alg{S}_2$, all vertical arrows arise from inclusion mappings and all horizontal arrows are defined by application of the corresponding functors to $f$. The triangle on the right commutes because $\func{T}'(X_2,\val{v}_2)$ is a cowedge and the two smaller rectangles commute by functoriality of $\func{T}$.
    
We define $\func{L}'$ on morphisms in a similar manner by duality. Since the definitions of $\func{T}'$ and $\func{L}'$ are completely dual, the following is obvious. 

\begin{thm}\label{thm:LiftingDualities}
If $\func{T}\colon \Stone \to \Stone$ and $\func{L}\colon \BA \to \BA$ are dual (that is, $\func{L}\cong \Pi\func{T}\Sigma$), then the corresponding liftings $\func{T}'\colon \StoneD \to \StoneD$ and $\func{L}'\colon \var{A}\to \var{A}$ are dual as well (that is, $\func{L}'\cong \Pi'\func{T}'\Sigma'$).  
\end{thm} 

For example, the `semi-primal Jónsson-Tarski duality' due to Maruyama \cite{Maruyama2012} can be obtained from the (usual) Jónsson-Tarski duality by this method. To illustrate this, we first show that there is an easier description of $\func{T}'$, given that $\func{T}$ \emph{preserves mono- and epimorphisms}. If $\func{T}$ preserves monomorphisms, then it also preserves finite intersections in the sense of \cite[Proposition 2.1]{Trnkova1969} (the proof therein still works for $\Stone$ instead of $\Set$, which is easy to check). It also preserves restrictions of morphisms $f\colon X \to Y$ to $X_0 \subseteq X$ in the sense that $\func{T}(f{\mid}_{X_0}) = (\func{T}f){\mid}_{\func{T}X_0}$ if we identify $\func{T}X_0$ with a subset of $\func{X}$. If $\func{T}$ also preserves epimorphisms, then $\func{T}$ preserves images in the sense that $\func{T}(f(X)) = \func{T}f(\func{TX})$ for every $f\colon X \to Y$ (since $\Stone$ is a regular category in which all epimorphisms are regular).  

\begin{prop}\label{prop:LiftingStoneMonoEpi}
Let $\func{T}\colon \Stone \to \Stone$ preserve mono- and epimorphisms. Then the following functor $\hat{\func{T}}\colon \StoneD \to \StoneD$ is naturally isomorphic to the lifting $\func{T}'$. On objects $(X,\val{v})$, the functor $\hat{\func{T}}$ is defined by 
$$
\hat{\func{T}}(X,\val{v}) = (\func{T}(X), \hat{\val{v}}),
$$
where, for $Z\in \func{T}(X)$, considering $\func{T}\CS(X,\val{v})$ as subspace of $\func{T}(X)$, 
$$
\hat{\val{v}}(Z) = \bigwedge \{ \alg{S} \mid Z \in \func{T}\C^{\alg{S}}(X,\val{v}) \}.
$$
On morphisms, $\hat{\func{T}}$ acts precisely like $\func{T}$.   
\end{prop} 

\begin{proof}
First note that $\hat{T}$ is well-defined on objects because $\CS\hat{\func{T}}(X,\val{v}) \cong \func{T}\CS(X,\val{v})$ and $\func{T}$ preserves intersections as discussed in the paragraph before the proposition. It is also well-defined on morphisms since, given $f\colon (X_1,\val{v}_1)\to (X_2,\val{v}_2)$, we have 
$$
(\func{T}f) (\func{T}\C^\alg{S}(X_1,\val{v}_1)) = \func{T} f(\C^\alg{S}(X_1,\val{v})) \subseteq \func{T} \C^\alg{S}(X_2,\val{v}_2),
$$
where in the first step we used that $\func{T}$ preserves images and restrictions as described above and in the second step we used that $\func{T}$ preserves inclusions up to isomorphisms and $f(\CS(X_1,\val{v}_1)) \subseteq \CS(X_2,\val{v}_2)$, which holds because $f$ is a morphism in $\StoneD$.

Next we show that $\hat{\func{T}}(X,\val{v})$ with the inclusions $i_\alg{S}\colon\VS \func{T} \CS (X,\val{v}) \hookrightarrow (\func{T}(X),\hat{\val{v}})$ is a cowedge over the diagram defining $\func{T}'$. Since $\func{T}$ preserves monomorphisms, we only need to make sure that the $i_\alg{S}$ are $\StoneD$-morphisms. But this is easy to see, since for $Z\in \func{T}\CS(X,\val{v})$ we have 
$$
\hat{\val{v}}(Z) = \bigwedge \{ \alg{R} \in \mathbb{S}(\alg{D}) \mid Z\in \func{T}\C^\alg{R}(X,\val{v}) \} \leq \alg{S}.
$$
For universality, simply note that, given another cowedge $c_\alg{S}\colon\VS \func{T} \CS (X,\val{v}) \to (Y,\val{w})$, the underlying map of $c_\alg{D}\colon \func{T}(X)\to Y$ also defines a morphism $\hat{\func{T}}(X,\val{v}) \to (Y,\val{w})$, and this morphism witnesses the universal property (due to an argument similar to the one in the proof of Proposition~\ref{prop:CoendPresentation}). 

The diagram after Definition~\ref{def:canonicalLiftingtop} was used to define $\func{T}'$ on morphisms $f\colon (X,\val{v}) \to (Y,\val{w})$. Since the diagram 
$$
\begin{tikzcd}[row sep=4.5em,column sep=4.5em]
\VS\func{T}\CS(X,\val{v}) \arrow[r, "\func{T}f{\mid}_{\CS(X,\val{v})}"] \arrow[d, hookrightarrow] & \VS\func{T}\CS(Y,\val{w}) \arrow[d, hookrightarrow] \\
\hat{\func{T}}(X,\val{v}) \arrow[r, "\func{T}f"'] &  \hat{\func{T}}(Y,\val{w})
\end{tikzcd}
$$            
commutes for every $\alg{S}\in \mathbb{S}(\alg{D})$, the morphisms $\func{T}'f$ and $\hat{\func{T}}f$ coincide by uniqueness.
\end{proof}
Due to this proposition, from now on we may not distinguish between $\hat{\func{T}}$ and $\func{T}'$ in our notation if $\func{T}$ preserves mono- and epimorphisms.

In particular, Proposition~\ref{prop:LiftingStoneMonoEpi} applies if $\func{T}$ is the dual of a functor $\func{L}\colon \BA \to \BA$ appearing in a concrete coalgebraic logic, that is, if $\func{L}$ has a presentation by operations and equations. 

\begin{cor}\label{cor:NiceLiftingifPresentation}
If $\func{L}\colon \BA \to \BA$ has a presentation by operations and equations, then it preserves mono- and epimorphisms. Therefore, for $\func{T} = \Sigma\func{L}\Pi$ the dual of $\func{L}$, the lifting $\func{T}'$ can be obtained as in Proposition~\ref{prop:LiftingStoneMonoEpi}. 
\end{cor}

\begin{proof}
A functor $\func{L}$ has a presentation by operations and equations if and only if it preserves sifted colimits. Since every filtered colimit is sifted, $\func{L}$ is finitary and preserves monomorphisms due to \cite[Lemma 6.14]{KurzPetrisan2010}. 

If $e\colon \alg{B} \to \alg{C}$ is a (necessarily regular) epimorphism between Boolean algebras, then $\alg{C}$ is isomorphic to a quotient of $\alg{B}$ by a congruence (namely, the kernel of $e$). Such a quotient is a reflexive coequalizer, which is preserved by $\func{L}$. Therefore, $\func{L}e\colon \func{L}(\alg{B}) \to \func{L}(\alg{C})$ is a coequalizer in $\BA$, which implies that $\func{L}e$ is an epimorphism again. 
\end{proof}

Next we discuss two examples of dualities between a category of coalgebras on $\Stone$ and a category of algebras on $\BA$ which can be lifted in this way.   

\begin{exa}\label{exam:LiftingVietoris}
The \emph{Vietoris functor} $\mathcal{V}\colon \Stone \to \Stone$ is dual to the functor $\func{L}$ from Example~\ref{exam:ModalAlgebras}, whose algebras correspond to modal algebras (this algebra-coalgebra version of Jónsson-Tarski duality is due to \cite{KupkeKurzVenema2003}). By Proposition~\ref{prop:LiftingStoneMonoEpi}, it is easy to see that the lifting $\mathcal{V}'$ of $\mathcal{V}$ coincides with the functor described by Maruyama in \cite[Section 4]{Maruyama2012}. 

The category of coalgebras $\Coalg{\mathcal{V}'}$ is isomorphic to (see also Example~\ref{exam:LiftingPowerset}) the category of \emph{descriptive general $\alg{D}$-frames}, which are triples $(X,R,\val{v})$ where 
\begin{itemize}
\item $(X,\val{v}) \in \StoneD$, 
\item $(X,R)$ is a descriptive general frame, 
\item if $x_1Rx_2$ then $v(x_2) \leq v(x_1)$ (Compatibility). 
\end{itemize}
A morphism of descriptive general $\alg{D}$-frames $(X_1,R_1,\val{v}_1) \to (X_2,R_2,\val{v}_2)$ is a map $f\colon X_1 \to X_2$ which is both a $\StoneD$-morphism and a bounded morphism. 

By Theorem~\ref{thm:LiftingDualities}, the category $\Coalg{\mathcal{V}'}$ of descriptive general $\alg{D}$-frames is dually equivalent to the category $\Alg{\func{L}'}$ of algebras for the lifting of $\func{L}$ to $\var{A}$. Later on, we provide a concrete presentation of $\func{L}'$ by operations and equations (see Subsection~\ref{subsubs:classicalModalLogic}).    
\end{exa}

\begin{exa}\label{exam:DescriptiveNeighborhoodLifting}
A \emph{descriptive neighborhood frame} is a coalgebra for the functor $\var{D}\colon \Stone \to \Stone$ defined in \cite[Subsection 5.1]{BezhanishviliDeGroot2022}. Therein, it is also shown that it is the dual of the functor $\func{L}\colon \BA \to \BA$ which has a presentation by one unary operation $\bigtriangleup$ and the empty set of equations (this is also known as Došen duality due to \cite{Dosen1989}).

Therefore, we can use Proposition~\ref{prop:LiftingStoneMonoEpi} to find the lifting $\mathcal{D'}$of $\mathcal{D}$. The corresponding coalgebras are \emph{descriptive neighborhood $\alg{D}$-frames}, \emph{i.e.}, tuples $(X,(N_x)_{x\in X},\val{v})$ where 
\begin{itemize}
\item $(X,\val{v}) \in \StoneD$,
\item $(X, (N_x)_{x\in X})$ is a descriptive neighborhood frame,
\item If $\val{v}(x) = \alg{S}$, then there is a collection of clopens $M\subseteq \mathcal{P}(\CS(X,\val{v}))$ such that 
$$
Y \in N_x \Leftrightarrow Y\cap \CS(X,\val{v}) \in M.
$$
\end{itemize}
Again, a morphism of descriptive neighborhood $\alg{D}$-frames is a map which is both a $\StoneD$-morphism and a coalgebra morphism.     
\end{exa}

Our next goal is to show that if $\func{L}\colon \BA \to \BA$ has a presentation by operations and equations, the same is true for its lifting $\func{L}'$ (for now, we only focus on the existence of such a presentation). For this, we will need the following.   

\begin{lem}\label{lem:L'siftedcolimits}
Let $\func{L}\colon \BA \to \BA$ be a functor and let $\func{T} = \Sigma\func{L}\Pi$ be its dual. 
\begin{enumerate}[(i)]
\item If $\func{L}$ is finitary (\emph{i.e.}, preserves filtered colimits), then so is $\func{L}'$.
\item If $\func{L}$ preserves mono- and epimorphisms, then $\func{T}'$ preserves all equalizers which $\func{T}\U$ preserves.  
\end{enumerate} 
\end{lem}   

\begin{proof}
(i): The functor $\KS$ preserves all colimits because it is a left-adjoint, $\func{L}$ is finitary by assumption and it is shown as in the proof of \cite[Theorem 4.11]{KurzPoigerTeheux2023} that $\PS$ is finitary. Since, in addition, filtered colimits commute with finite limits, this implies that $\func{L}'$ is finitary as well.

(ii): Let $f,g\colon (X_1,\val{v}_1) \to (X_2,\val{v}_2)$ be two $\StoneD$-morphisms. It is easy to check that the equalizer of $f$ and $g$ is given by $(E,\val{w})$ where $E\subseteq X_1$ is the equalizer of $f$ and $g$ in $\Stone$ and $\val{w}$ is the restriction of $\val{v}_1$ to $E$. Now, assuming that $\func{T}(E)$ is the equalizer of $\func{T}f$ and $\func{T}g$ in $\Stone$, we show that $(\func{T}(E), \hat{\val{w}})$ (in the notation of Proposition~\ref{prop:LiftingStoneMonoEpi}) is the corresponding equalizer in $\StoneD$. By definition, for $Z\in \func{T}(E)$ we have 
\begin{align*}
\hat{\val{w}}(Z) & = \bigwedge\{ \alg{S} \mid Z \in \func{T}\CS(E,\val{w}) \} \\
 &    = \bigwedge\{ \alg{S} \mid Z\in \func{T}\CS(X,\val{v}_1) \cap \func{T}(E) \} \\
 & = \bigwedge \{ \alg{S} \mid Z\in \func{T}\CS(X,\val{v}_1) \} = \hat{\val{v}}_1(Z), 
\end{align*}    
where we used that $\func{T}$ preserves finite intersections and $\CS(E,\val{w}) = \CS(X,\val{v}_1)\cap E$ since $\val{w}$ is the restriction of $\val{v}_1$. Thus $\hat{\val{w}}$ is the restriction of $\hat{\val{v}}_1$ to $\func{T}(E)$, finishing the proof. 
\end{proof}

We can now easily conclude the following. 

\begin{cor}\label{cor:L'haspresentation}
If $\func{L}\colon \BA \to \BA$ has a presentation by operations and equations, then the same is true for the lifting $\func{L}'\colon \var{A} \to \var{A}$. 
\end{cor}

\begin{proof}
By \cite[Theorem 4.7]{KurzRosicky2012} we know that $\func{L}'$ has a presentation if and only if it preserves sifted colimits. By \cite[Theorem 7.7]{AdamekRosickyVitale2010} we know that $\func{L}'$ preserves sifted colimits if and only if it preserves filtered colimits and reflexive coequalizers. Since $\func{L}$ has a presentation, it preserves filtered colimits and reflexive coequalizers. By part (i) of Lemma~\ref{lem:L'siftedcolimits} we know that $\func{L}'$ preserves filtered colimits and by part (ii) we know that $\func{L}'$ preserves reflexive coequalizers since $\func{T}'$ preserves coreflexive equalizers (because $\func{T}\U$ preserves them as well).  
\end{proof}

Thus we showed that a presentation of $\func{L}'$ necessarily \emph{exists}, however our proof does not indicate what such a presentation actually `looks like'.
In Section~\ref{sec:MVConcrete}, we describe some circumstances under which one can directly obtain a presentation of $\func{L}'$ from a given presentation of $\func{L}$. 

But before we deal with these matters, which essentially concern \emph{concrete} coalgebraic logics, in the following section we first describe how to lift classical \emph{abstract} coalgebraic logics to the semi-primal level.  
\section{Many-valued abstract coalgebraic logic}\label{sec:ManyvaluedCoalgModalLogic} 
In this section, we discuss how to lift classical abstract coalgebraic logics $(\func{L},\delta)$ (recall Definition~\ref{defin:AbstractCoalgebraicLogic}) to many-valued abstract coalgebraic logics $(\func{L}',\delta')$, where $\func{L}'$ is an endofunctor on the variety $\var{A}$.  

We will be more flexible in our choice of base category on the semantical side. It is common in many-valued modal logic (see, \emph{e.g.}, \cite{Bou2011}) to interpret formulas on Kripke frames, making $\Set$ a natural candidate for this base category. However, in light of the previous section, a category closer to $\StoneD$ seems preferable from a duality-theoretical point of view. Our compromise is the category $\SetD$, which provides a natural environment for coalgebraic logics over $\var{A}$. More `familiar' results on many-valued coalgebraic logics for functors $\func{T}\colon \Set\to \Set$ are immediate consequences of the corresponding results for the lifted functors $\func{T}'\colon \SetD \to \SetD$ (see Corollary~\ref{cor:ReplaceT'byT}). In the case where $\alg{D}$ is primal, the corresponding results already appeared in \cite{KurzPoiger2023} by the first two authors.     
\subsection{Lifting Endofunctors on Set}\label{subsec:LiftingFunctors} Just like $\Set$ arises from $\Stone$ by `forgetting topology', the following category arises from $\StoneD$ in this sense.   
\begin{defi}\label{defin:SetL}
The category $\SetD$ has objects $(X,v)$, where $X$ is a set and $v\colon X \to \mathbb{S}(\alg{D})$ is a map assigning a subalgebra of $\alg{D}$ to every point of $X$. A morphism $f\colon (X_1,v_1) \to (X_2, v_2)$ is a map $X_1 \to X_2$ which satisfies 
$$ v_2(f(x)) \leq v_1(x) $$
for every $x\in X$.  
\end{defi}
We define functors $\func{S}'\colon \var{A} \to \SetD$ and $\func{P}'\colon \SetD\to \var{A}$ (similar to the functors $\Sigma'$ and $\Pi'$ defined in Subsection~\ref{subs:SemiPrimalDuality}) which form a dual adjunction between $\SetD$ and $\var{A}$.

The functor $\func{S}'\colon \var{A} \to \SetD$ is defined on objects $\alg{A}\in \var{A}$ by 
$$ \func{S}'(\alg{A}) = \big(\var{A}(\alg{A},\alg{D}), \im\big),$$           
where every homomorphism $u\colon \alg{A} \to \alg{D}$ gets assigned its image $\im(u) = u(A) \leq \alg{D}$. To a homomorphism $h\colon \alg{A} \to \alg{A'}$ the functor $\func{S}'$ assigns the $\SetD$-morphism 
$\func{S}' h \colon \var{A}(\alg{A}',\alg{D}) \to \var{A}(\alg{A},\alg{D})$ given by 
$u \mapsto u \circ h. $

The functor $\func{P}' \colon \SetD \to \var{A}$ is defined on objects by 
$$
\func{P}'(X,v) = \prod_{x\in X}v(x).
$$
To a morphism $f\colon (X_1,v_1) \to (X_2,v_2)$, the functor $\func{P}'$ assigns the homomorphism $\alpha \mapsto \alpha \circ f$. 
These functors define a dual adjunction. The corresponding natural transformations $\eta'\colon 1_\var{A} \Rightarrow \func{P'S'}$ and $\varepsilon'\colon 1_{\SetD} \Rightarrow \func{S'P'}$ are given by evaluations, that is, for for all $\alg{A}\in \var{A}$ and $(X,v)\in \SetD$ we have 
\begin{align*}
\eta'_\alg{A}\colon \alg{A} &\to \prod_{u\in \var{A}(\alg{A},\alg{D})}\im(u)  &  \varepsilon'_{(X,v)}\colon X &\to \var{A}\big(\prod_{x\in X}v(x),\alg{D}\big) \\
a &\mapsto \mathsf{ev}_a & x &\mapsto \mathsf{ev}_x
\end{align*}
where $\mathsf{ev}_a(u) = u(a)$ and $\mathsf{ev}_x(\alpha) = \alpha(x)$. Note that if $\alg{D} = \alg{2}$ is the two-element Boolean algebra, this adjunction coincides with the one given by $\func{S}$ and $\func{P}$ in Example~\ref{exam:classicalModalLogic}. We will keep this notation and also use $\eta$ and $\varepsilon$ instead of $\eta'$ and $\varepsilon'$ in this case.  

Analogous to Section~\ref{sec:LiftingDualities} (see Figure~\ref{fig:1}), there exists an adjunction for every subalgebra $\alg{S} \in \mathbb{S}(\alg{D})$. Slightly abusing notation, we again denote the functors involved by $\VS\colon \Set \to \SetD$ and $\CS\colon \SetD \to \Set$ (see Figure~\ref{fig:2}). 
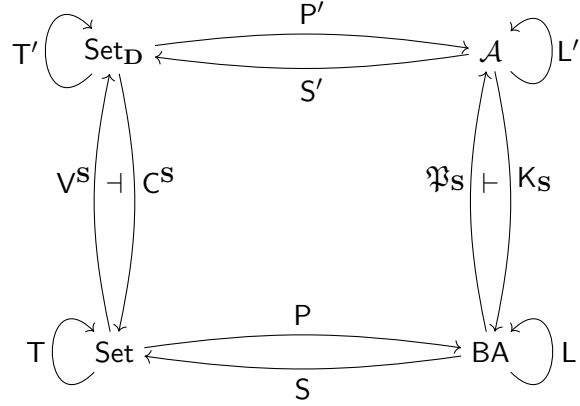
\begin{figure}[ht]
$$
\begin{tikzpicture}
  \node (Stone) at (-2.5,0) {$\Set$};
  \node (BA) at (2.5,0) {$\BA$};
  \node (StoneD) at (-2.5,4) {$\SetD$};
  \node (A) at (2.5,4) {$\var{A}$};
  \node (lad) at (-2.5,2.25) {$\dashv$};
  \node (rad) at (2.5,2.25) {$\vdash$};
  \draw [->, bend left = 8] (Stone) to node[above]{$\func{P}$}  (BA);
  \draw [->, bend left = 8] (BA) to node[below]{$\func{S}$}  (Stone);
  \draw [->, bend left = 8] (StoneD) to node[above]{$\func{P}'$}  (A);
  \draw [->, bend left = 8] (A) to node[below]{$\func{S}'$}  (StoneD);
  \draw [->, bend left = 12] (Stone) to node[above]{$\VS\phantom{.....}$}  (StoneD);
  \draw [->, bend left = 12] (StoneD) to node[above]{$\phantom{......}\CS$}  (Stone);
  \draw [->, bend left = 12] (BA) to node[above]{$\PS\phantom{......}$}  (A);
  \draw [->, bend left = 12] (A) to node[above]{$\phantom{......}\KS$}  (BA);
  \draw [->,out=225,in=135,looseness=5] (Stone) to node[pos = 0.35,above]{$\func{T}\phantom{.....}$}  (Stone);
  \draw [->,out=225,in=135,looseness=5] (StoneD) to node[pos = 0.35,above]{$\func{T}'\phantom{......}$}  (StoneD);
  \draw [->,out=-45,in=45,looseness=5] (BA) to node[pos = 0.35,above]{$\phantom{.....}\func{L}$}  (BA);
  \draw [->,out=-45,in=45,looseness=5] (A) to node[pos = 0.35,above]{$\phantom{.....}\func{L}'$}  (A);
\end{tikzpicture}
$$
\caption{Lifting coalgebraic logics via the subalgebra adjunctions.}
\label{fig:2}
\end{figure} 

In the case where $\alg{S} = \alg{D}$ is the entire algebra, the functor $\CS$ is the forgetful functor and will be denoted by $\U\colon \SetD \to \Set$. On the algebraic side, $\KS$ is then the Boolean skeleton $\Skel\colon \var{A} \to \BA$ which was already defined in Section~\ref{sec:LiftingDualities}.

Of course, there is no longer a dual equivalence between the left-hand and right-hand side of Figure~\ref{fig:2} (however, restricted to the finite level there still is). Fortunately, the following useful relationships still hold. 

\begin{lem}\label{lem:SetDAdjunctionsandBAAdjunctions}
The functors involved in Figure~\ref{fig:2} and the natural transformations $\varepsilon, \eta, \varepsilon', \eta'$ satisfy the following properties.
\begin{enumerate}[(i)]
\item For all $\alg{S}\in \mathbb{S}(\alg{D})$, there is a natural isomorphism $\Theta^{\alg{S}}\colon \func{P'}\VS\Rightarrow \PS\func{P}$, with components 
$$
\Theta_X^\alg{S}(\alpha)(s) = T_s(\alpha),
$$
where $\alpha\in \func{P'}\VS(X) \cong \alg{S}^X$ and $T_s$ are the unary terms defined in Theorem~\ref{thm:semiprimal_characterizations}. 
\item For all $\alg{S}\in \mathbb{S}(\alg{D})$, there is a natural isomorphism $\Psi^{\alg{S}}\colon \func{P}\CS\Rightarrow \KS\func{P'}$ given by the identification of $\alg{2}^{\CS(X,v)}$ with 
$$
\KS\big(\prod_{X} v(x)\big) \cong \Skel \big(\prod_{\CS(X,v)} v(x)\big) \cong \prod_{ \CS(X,v)} \Skel\big(v(x)\big) \cong \alg{2}^{\CS(X,v)}.
$$
In particular, for $\alg{S} = \alg{D}$, there is a natural isomorphism $\Psi\colon \func{P}\U\Rightarrow \Skel\func{P'}$.
\item There is a natural isomorphism $\Phi\colon \U\func{S}' \Rightarrow \func{S}\Skel$ given by restriction
\begin{align*}
\Phi_\alg{A}\colon \var{A}\big(\alg{A},\alg{D}\big) &\to \BA\big(\Skel(\alg{A}),\alg{2}\big) \\
u &\mapsto u{\mid}_{\Skel(\alg{A})}.
\end{align*}   
\item For $\Psi$ and $\Phi$ from (ii) and (iii) the identities $$\varepsilon\U = \func{S}\Psi \circ \Phi\func{P}' \circ \U\varepsilon'\text{ and } \Skel\eta' = \Psi\func{S}' \circ \func{P}\Phi \circ \eta\Skel
$$ hold.
\end{enumerate} 
\end{lem}

\begin{proof}
For part (i), showing that $\Theta_X$ is a homomorphism is analogous to \cite[Proposition 4.8]{KurzPoigerTeheux2023}. It is injective because $\alpha_1 \neq \alpha_2$ means there is some $x\in X$ such that $\alpha_1(x) \neq \alpha(x)$. Then, for $s = \alpha(x)$ we have $T_s(\alpha_1) \neq T_s(\alpha_2)$. 

To see that $\Theta_X$ is surjective, let $\xi\colon S \to 2^X$ be in $\PS \func{P}(X)$. By definition of the Boolean power this means that, in every component $x\in X$, there is a unique $s_x\in \alg{S}$ with $\xi(s_x)(x) = 1$. Thus $\alpha(x) = s_x$ is in the preimage of $\xi$. 
  
For naturality, we need to show that the following diagram commutes for any morphism $f\colon Y \to X$.
$$
\begin{tikzcd}[row sep=4.5em,column sep=4.5em]
\func{P'}\VS(X) \arrow[r, "\Theta_X"] \arrow[d, "\func{P'}\VS f"'] & \PS\func{P}(X) \arrow[d, "\PS\func{P}f"] \\
\func{P'}\VS(Y) \arrow[r, "\Theta_Y"'] &  \PS\func{P}(Y)
\end{tikzcd}
$$             
Given $\alpha\in \func{P}'\VS(X)\cong \alg{S}^X$, on the one hand we have $\PS\func{P}f(\theta_X(\alpha)) = \func{P}f \circ \Theta_X(\alpha)$, which sends $s\in S$ to $\func{P}f(T_s(\alpha)) = T_s(\alpha \circ f)$. On the other hand we have $\Theta_Y(\func{P'}\VS f(\alpha)) = \Theta_Y(\alpha \circ f)$ sends $s$ to $T_s(\alpha \circ f)$ as well. This finishes the proof of part (i).

The equations in (ii) follow from the results of \cite[Subsection 4.4]{KurzPoigerTeheux2023}, the fact that $\Skel$ preserves limits and $\Skel(\alg{S}) \cong \alg{2}$ holds for all $\alg{S}\in \mathbb{D}$. Naturality is easy to check by definitions. 

The proof of (iii) is completely analogous to that of \cite[Proposition 4.3]{KurzPoigerTeheux2023} and the proof of (iv) is completely analogous to that of \cite[Proposition 10(c)]{KurzPoiger2023}.
\end{proof}    

Exactly like in Definition~\ref{def:canonicalLiftingtop}, we can use these subalgebra adjunctions to lift an endofunctor $\func{T}\colon \Set \to \Set$ to $\func{T}'\colon \SetD\to \SetD$. 

\begin{defi}\label{def:canonicalLiftingSetD}
Let $\func{T}\colon \Set \to \Set$ be a functor. 
The \emph{lifting of $\func{T}$ to $\SetD$} is the functor $\func{T}'\colon \SetD \to \SetD$ defined on objects by
$$
\func{T}'(X,v) = \int^{\alg{S}\in\mathbb{S}(\alg{D})} \VS \func{T} \CS (X,v)
$$
and on morphisms as discussed in the paragraph after Definition~\ref{def:canonicalLiftingtop}.   
\end{defi}                 

Again, the canonical lifting $\func{T}'$ of $\func{T}$ can be described more concretely if the functor $\func{T}$ preserves mono- and epimorphisms. In particular, this is true for $\Set$-endofunctors which are \emph{standard} (that is, inclusion-preserving), and up to what it does on the empty set, every $\Set$-endofunctor is naturally isomorphic to one which is standard \cite{Trnkova1969}. The proof of Proposition~\ref{prop:LiftingStoneMonoEpi} can be adapted to obtain the following.

\begin{prop}\label{prop:LiftingSetStandard}
Let $\func{T}\colon \Set \to \Set$ preserve mono- and epimorphisms. Then, up to natural isomorphism, $\func{T}'$ is defined on objects by 
$$
\func{T}'(X,v) = (\func{T}(X), \hat{v}),
$$
where, for $Z\in \func{T}(X)$, considering $\func{T}\CS(X,v)$ as subspace of $\func{T}(X)$, 
$$
\hat{v}(Z) = \bigwedge \{ \alg{S} \mid Z \in \func{T}\C^{\alg{S}}(X,v) \}.
$$
On morphisms, $\func{T}'$ acts precisely like $\func{T}$.   
\end{prop} 

An obvious question related to modal logic is what the lifting of the powerset functor $\mathcal{P}\colon \Set \to \Set$ and the corresponding coalgebras look like.

\begin{exa}\label{exam:LiftingPowerset}
Since $\mathcal{P}$ is standard, by Proposition~\ref{prop:LiftingSetStandard} its lifting is given on objects by $\mathcal{P}'(X,v) = (\mathcal{P}(X),\hat{v})$, where, for $Y\subseteq X$ we have
$$
\hat{v}(Y) = \bigwedge\{ \alg{S} \mid Y\in \mathcal{P}\C^\alg{S}(X,v) \} = \bigvee \{ v(y) \mid y\in Y \},$$
where the second equation holds because both terms describe the smallest subalgebra $\alg{S}$ which satisfies $v(y) \leq \alg{S}$ for all $y\in Y$. 

Similarly to Example~\ref{exam:LiftingVietoris}, we define a \emph{Kripke $\alg{D}$-frame} to be a triple $(X,R,v)$ such that  
\begin{itemize}
\item $(X,R)$ is a Kripke frame,
\item $(X,v)$ is a member of $\SetD$, 
\item if $x_1Rx_2$ then $v(x_2) \leq v(x_1)$ (Compatibility).
\end{itemize}
A \emph{bounded $\alg{D}$-morphism} between Kripke $\alg{D}$-frames $(X_1,R_1,v_1)$ and $(X_2,R_2,v_2)$ is a map $f\colon X_1\to X_2$ which is both a bounded morphism $(X_1,R_1)\to (X_2,R_2)$ and a $\SetD$-morphism $(X_1,v_1) \to (X_2,v_2)$. 

The category thus arising is isomorphic to $\Coalg{\mathcal{P}'}$. Indeed, given a $\mathcal{P}'$-coalgebra $\gamma\colon (X,v) \to \mathcal{P}'(X,v)$, we define the corresponding Kripke $\alg{D}$-frame $(X,R_\gamma, v)$ by $x_1 R_\gamma x_2 \Leftrightarrow x_2 \in \gamma(x_1)$. The compatibility condition is satisfied because $\gamma$ being a $\SetD$-morphism implies 
$$ \hat{v}(\gamma(x)) = \bigvee \{ v(x') \mid x'\in \gamma(x) \} \leq v(x).$$

Conversely, to every Kripke $\alg{D}$-frame $(X,R,v)$ we can associate the $\mathcal{P}'$-coalgebra $\gamma_R \colon (X,v) \to \mathcal{P}'(X,v)$ given by $\gamma_R(x) = R[x]$. Again, the compatibility condition is what is needed to assure that this is a morphism in $\SetD$. 

In the case where $\alg{D} = \lucas_n$ is a \L ukasiewicz-chain, Kripke $\alg{D}$-frames appeared first in \cite[Section 7]{HansoulTeheux2013} (called `\emph{$\lucas_n$-valued frames}' therein).

The intended semantics of modal formulas over a Kripke $\alg{D}$-frames $(X,R,v)$ are the following. Given a set $\mathsf{Prop}$ of propositional variables, the corresponding valuations will only be those $\mathsf{Val}\colon X \times \mathsf{Prop} \to D$ which satisfy $\mathsf{Val}(x,p) \in v(x)$ for all $x\in X$. Such a valuation can be inductively extended to formulas obtained via connectives of $\alg{D}$ in the obvious way and to formulas of the form $\Box \varphi$ by
$$ \mathsf{Val}(x,\Box\varphi) = \bigwedge \{ \mathsf{Val}(x', \varphi) \mid xRx'\},$$
which will still satisfy $\mathsf{Val}(x,\Box\varphi) \in v(x)$.   
\end{exa}  

We come back to this example in Subsection~\ref{subsec:applications}, when we deal with concrete coalgebraic logics. In the following subsection, we describe how to lift abstract coalgebraic logics, and show that one-step completeness and expressivity are preserved under this lifting.

\subsection{Lifting Coalgebraic Logics, Completeness and Expressivity}\label{subsec:LiftingCoalgLogics}  

Assume we are given a classical abstract coalgebraic logic $(\func{L}, \delta)$ (as in Definition~\ref{defin:AbstractCoalgebraicLogic}) for $\func{T}\colon \Set \to \Set$. Since by now we have a way to lift $\func{T}$ and $\func{L}$, we only need to lift the natural transformation $\delta\colon \func{LP}\Rightarrow \func{PT}$ to a natural transformation $\delta'\colon \func{L'P'}\Rightarrow\func{P'T'}$. 

By definition we have 
$$\func{L'P'}(X,v) = \int_{\alg{S}\in\mathbb{S}(\alg{D})}\PS \func{L} \KS \func{P'}(X,v)$$
and by Lemma~\ref{lem:SetDAdjunctionsandBAAdjunctions}(ii) there is a natural isomorphism $\PS \func{L} \KS \func{P'} \cong \PS \func{LP} \CS$. Furthermore, we have 
$$ \func{P'T'}(X,v) = \func{P}' \big(\int^{\alg{S}\in \mathbb{S}(\alg{D})} \VS \func{T} \CS(X,v)\big) \cong \int_{\alg{S}\in\mathbb{S}(\alg{D})} \func{P'}\VS \func{T} \CS(X,v)$$
because $\func{P}$ is right-adjoint as a functor $\SetD^\mathrm{op}\to \var{A}$ and due to Lemma~\ref{lem:SetDAdjunctionsandBAAdjunctions}(i) we know $\func{P'}\VS\func{T}\CS \cong \PS \func{P}\func{T}\CS$. Because $\delta$ is natural, for every $(X,v)\in\SetD$ we can define a wedge
$$
\begin{tikzcd}
\func{L'P'}(X,v) \arrow[r] & \PS \func{LP} \CS(X,v) \arrow[r] & \PS \func{PT} \CS(X,v) \arrow[r] & \func{P'}\VS\func{T}\CS(X,v) 
\end{tikzcd}
$$
where the first arrow is the corresponding limit morphism up to the first natural isomorphism mentioned above, the second arrow is $\PS\delta_{\CS(X,v)}$ and the last arrow is the second natural isomorphism mentioned above.  Thus the universal property of the end $\func{P'T
'}(X,v)$ yields a morphism
$$ \delta'_{(X,v)}\colon \func{L'P'}(X,v) \to \func{P'T'}(X,v),$$
which in fact defines a natural transformation $\func{L'P'}\Rightarrow \func{P'T'}$, by naturality of $\delta$ and of all isomorphisms involved in the definition of $\delta'$.    

Now we have everything at hand to define the lifting of a classical abstract coalgebraic logic.

\begin{defi}\label{def:LiftingAbstractCoalgLogic}
Let $(\func{L},\delta)$ be a classical abstract coalgebraic logic for $\func{T}$. The \emph{lifting of $(\func{L},\delta)$ to $\var{A}$} is the abstract coalgebraic logic $(\func{L',\delta'})$ for $\func{T}'$, where $\func{L}'$ and $\func{T}'$ are the liftings of $\func{L}$ and $\func{T}$ to $\var{A}$ and $\SetD$, respectively, and $\delta'$ is the natural transformation defined above.  
\end{defi} 

In the remainder of this section, we show that under the assumption that $\func{L}$ preserves mono- and epimorphisms (in particular, if the coalgebraic logic is concrete), one-step completeness and expressivity of abstract coalgebraic logics are preserved under this lifting. In Section~\ref{sec:MVConcrete} we deal with concrete coalgebraic logics, in particular we discuss how to lift axiomatizations (\emph{i.e.}, presentations of functors) as well. 

First we deal with the preservation of one-step completeness (Definition~\ref{defin:oneStepCompleteness}) under the lifting of Definition~\ref{def:LiftingAbstractCoalgLogic}.

\begin{thm}\label{thm:OneStepComplPreservation}
Let $(\func{L},\delta)$ be a classical abstract coalgebraic logic for a functor $\func{T}\colon \Set \to \Set$ such that $\func{L}\colon \BA \to \BA$  and $\func{T}$ preserve mono- and epimorphisms. If $(\func{L},\delta)$ is one-step complete, then so is its lifting $(\func{L}',\delta')$.   
\end{thm}

\begin{proof}
By definition we have to show that if $\delta$ is a componentwise monomorphism, then so is $\delta'$. It suffices to show that $\Skel\delta' = \delta\U$ holds up to natural isomorphism since $\Skel$ is faithful and thus reflects monomorphisms. For all $(X,v)$, by our definition of $\delta'$ (in the special case where $\alg{S} = \alg{D}$) the following diagram commutes.
$$
\begin{tikzcd}[row sep=4.5em,column sep=4.5em]
\func{L'P'}(X,v) \arrow[r, "f"] \arrow[d, "\delta'_{(X,v)}"'] &
\Pow\func{L}\Skel\func{P'}(X,v) \arrow[r, "\cong"] &
\Pow\func{LPU}(X,v) \arrow[d, "\Pow\delta_X"]\\
\func{P'T'}(X,v) \arrow[r, "g"'] &
\func{P'V}^\alg{D}\func{TU}(X,v) \arrow[r,"\cong"] &
\Pow\func{PTU}(X,v)
\end{tikzcd}
$$
Here, $f$ and $g$ are defined via the corresponding limit morphisms. We now apply $\Skel$ to this diagram and use the fact that $\Skel\Pow \cong \id_{\BA}$ to get the following commutative diagram.
$$
\begin{tikzcd}[row sep=4.5em,column sep=4.5em]
\Skel\func{L'P'}(X,v) \arrow[r, "\Skel f"] \arrow[d, "\Skel \delta'_{(X,v)}"'] &
\Skel\Pow\func{L}\Skel\func{P'}(X,v) \arrow[r, "\cong"] &
\func{LPU}(X,v) \arrow[d, "\delta_X"]\\
\Skel\func{P'T'}(X,v) \arrow[r, "\Skel g"'] &
\Skel\func{P'V}^\alg{D}\func{TU}(X,v) \arrow[r,"\cong"] &
\func{PTU}(X,v)
\end{tikzcd}
$$ 
To conclude our proof, it remains to be shown that $\Skel f$ and $\Skel g$ are isomorphisms. 

The fact that $\Skel f$ commutes follows by duality from Proposition~\ref{prop:LiftingStoneMonoEpi}, where it is shown that the cowedge morphism corresponding to $\alg{S} = \alg{D}$ of the (Stone) dual of $\func{L'}$ is the identity on the underlying space, thus applying the forgetful functor (the Stone dual of $\Skel$) to it yields an isomorphism. 

Similarly, by Proposition~\ref{prop:LiftingSetStandard} (we also use the notation used there), again we know that the cowedge morphism $\func{V}^\alg{D}\func{T}\U(X,v) \to \func{T'}(X,v)$ is the identity map as a (notably non-identity) morphism $(\func{T}(X), v^\alg{D}) \to (\func{T}(X),\hat{v})$. The homomorphism $g$ is obtained by applying $\func{P}'$ to this morphism, and it is easy to see that this is the natural inclusion 
$g\colon \prod_{Z\in \func{T}(X)}\hat{v}(Z) \hookrightarrow \prod_{Z\in \func{T}(X)}\alg{D}$. Applying $\Skel$ to this natural inclusion clearly yields (up to identifying $\alg{2}$ with the subset $\{ 0,1 \} \subseteq D$) the identity $\alg{2}^{\func{T}(X)} \to \alg{2}^{\func{T}(X)}$. Thus $\Skel g$ is also an isomorphism, which concludes the proof.     
\end{proof} 

Similarly, expressivity (Definition~\ref{defin:expressivity}) is preserved under this lifting as follows.

\begin{thm}\label{thm:ExpressPreservation}
Let $(\func{L},\delta)$ be a classical abstract coalgebraic logic for a functor $\func{T}\colon \Set \to \Set$ such that $\func{L}\colon \BA \to \BA$ and $\func{T}$ preserve mono- and epimorphisms. If $(\func{L},\delta)$ is expressive, then so is its lifting $(\func{L}',\delta')$.   
\end{thm}

\begin{proof}
In light of Definition~\ref{defin:expressivity}, we note that by duality it can be seen that $\Alg{\func{L}'}$ has an initial object if $\Alg{\func{L}}$ has one (endow the terminal coalgebra of the dual of $\func{L}$ with the `bottom' evaluation which always assigns the smallest subalgebra of $\alg{D}$). Furthermore, $\SetD$ has epi-mono factorizations for essentially the same reason that $\Set$ does. 

We need to show that if the adjoint transpose $\delta^\dagger$ is a componentwise monomorphism, then so is the adjoint transpose $(\delta')^\dagger$. This works similarly to the proof in the primal case \cite[Theorem 12(c)]{KurzPoiger2023}. It suffices to show that $\U(\delta')^\dagger = \delta^\dagger\Skel$ holds up to natural isomorphism, since $\func{U}$ is fully faithful and thus reflects monomorphisms. In other words, we want to show that the following diagram commutes. 
$$
\begin{tikzcd}[row sep=4.5em,column sep=4.5em]
\func{UT'S'} \arrow[r, "\U\varepsilon'\func{T'S'}"] \arrow[d, ""{name=d1}] & \func{US'P'T'S'} \arrow[r, "\func{US'}\delta'\func{S'}"] \arrow[d, ""{name=d2}] &  \func{US'L'P'S'} \arrow[r, "\func{US'L'\eta'}"] \arrow[d, ""{name=d3}] & \func{S'L'} \arrow[d, ""{name=d4}] \\
\func{TS}\Skel \arrow[r, "\varepsilon\func{TS}\Skel"'] &  \func{SPTS}\Skel \arrow[r, "\func{S}\delta\func{S}\Skel"'] &  \func{SLPS}\Skel \arrow[r, "\func{SL}\eta\Skel"'] &  \func{SL}\Skel 
\arrow[phantom,from=d1,to=d2,"D_1"]
\arrow[phantom,from=d2,to=d3,"D_2"]
\arrow[phantom,from=d3,to=d4,"D_3"]
\end{tikzcd}
$$  
Here, the upper border of the (entire) diagram is $\U(\delta')^\dagger$ and the lower border is $\delta^\dagger\Skel$. All vertical arrows are natural isomorphisms obtained via the natural isomorphisms $\Psi\colon \func{PU}\Rightarrow \Skel\func{P'}$, $\Phi\colon \func{US'}\Rightarrow \func{S}\Skel$ from Lemma~\ref{lem:SetDAdjunctionsandBAAdjunctions}, the identity $\func{UT'} = \func{TU}$ (which clearly holds by Proposition~\ref{prop:LiftingSetStandard}) and natural isomorphism $\Skel\func{L'} \cong \func{L'}\Skel$ (which exists by the dual of Proposition~\ref{prop:LiftingStoneMonoEpi}).

The diagram $D_1$ commutes because, applying the first equation of Lemma~\ref{lem:SetDAdjunctionsandBAAdjunctions}(iv), we can compute
$$
\func{SPT}\Phi \circ \func{S}\Psi\func{T'S'} \circ \Phi\func{P'T'S'} \circ \U\varepsilon'\func{T'S'} = \func{SPT}\Phi \circ (\func{S}\Psi \circ \Phi\func{P}' \circ \U\varepsilon')\func{T'S'} =  \func{SPT}\Phi \circ \varepsilon\func{UT'S'},
$$
which coincides with $\varepsilon\func{TS}\Skel \circ \func{T}\Phi$.

The diagram $D_3$ commutes for similar reasons since, applying the second equation of Lemma~\ref{lem:SetDAdjunctionsandBAAdjunctions}(iv), we can compute
$$
\func{SL}\eta\Skel \circ \func{SLP}\Phi \circ \func{SL}\Psi\func{S'} \circ \Phi\func{L'P'S'} = \func{SL}(\Psi\func{S}' \circ \func{P}\Phi \circ \eta\Skel) \circ \Phi\func{L'P'S'} = \func{SL}\Skel\eta' \circ \Phi{\func{L'P'S'}},
$$
which coincides with $\Phi\func{L'} \circ \func{S'L'}\eta'$. 

Finally, to see that the diagram $D_2$ commutes, one uses $\Skel\delta' = \delta\U$ (up to natural isomorphisms), as shown in the proof of Theorem~\ref{thm:OneStepComplPreservation}.
\end{proof}

While $(\func{L}',\delta')$ is a coalgebraic logic for the lifting $\func{T}'$ of $\func{T}$, it also directly yields a coalgebraic logic for $\func{T}$ itself. Indeed, with the exception of \cite{HansoulTeheux2013}, all results on many-valued modal logic interpret formulas over Kripke frames (\emph{i.e.}, $\func{T}$-coalgebras) rather than Kripke $\alg{D}$-frames (\emph{i.e.}, $\func{T}'$-coalgebras). This is easily dealt with, since from $(\func{L}',\delta')$ we can always obtain a coalgebraic logic $(\func{L}', \delta^\top)$ for $\func{T}$ by composing with the adjunction $\V^\top \dashv \U$. That is, we simply define $\delta^\top\colon \func{L'P'}\Vtop \to \func{P'}\Vtop\func{T}$ to be $\delta'\Vtop$ (which is well-defined because $\func{T}'\Vtop = \Vtop \func{T}$). In the case $\func{T} = \mathcal{P}$, this essentially means that we identify a Kripke frame $(X,R)$ with the corresponding Kripke $\alg{D}$-frame $(X,R,v^\top)$ (which from a logical perspective means that models can have arbitrary valuations $\mathsf{Val}\colon X \times \mathsf{Prop} \to D$). It is obvious by definition that one-step completeness of $(\func{L}',\delta')$ implies one-step completeness of $(\func{L}',\delta^\top)$. Furthermore, the fact that $(\delta^\top)^\dagger = \U (\delta')^\dagger$ yields the analogous result for expressivity. Thus, together with Theorems~\ref{thm:OneStepComplPreservation} and \ref{thm:ExpressPreservation} we showed the following.

\begin{cor}\label{cor:ReplaceT'byT}
Let $(\func{L},\delta)$ be a classical abstract coalgebraic logic for $\func{T}$ as in Theorem~\ref{thm:OneStepComplPreservation}, let $(\func{L}',\delta')$ be its lifting and let $\delta^\top = \delta' \Vtop$. Then $(\func{L}',\delta^\top)$ is an abstract coalgebraic logic for $\func{T}$, which is one-step complete if $(\func{L},\delta)$ is one-step complete and expressive if $(\func{L},\delta)$ is expressive.   
\end{cor}     

In this section we showed that both one-step completeness and expressivity of classical coalgebraic logics are preserved under the lifting to the many-valued level. On the level of abstract coalgebraic logics this is satisfying, but from a more `practical' perspective these results only become interesting once we discuss concrete coalgebraic logics and provide an axiomatization of the lifted logic. This is the content of the next section, where we also explicitly show how our results apply to classical modal logic (Kripke semantics) and to neighborhood semantics.  
\section{Many-valued concrete coalgebraic logic}\label{sec:MVConcrete}
In this section, we deal with liftings of classical concrete coalgebraic logics. Most notably, we propose a method which, in some cases (including classical modal logic), allows us to find a presentation of $\func{L}'\colon \var{A} \to \var{A}$, given a presentation of $\func{L}\colon \BA \to \BA$. As in the previous section, the results here can be viewed as direct generalizations of those of the first two authors \cite{KurzPoiger2023} about the case where the algebra of truth-degrees $\alg{D}$ is assumed to be primal. In Subsection~\ref{subsec:applications}, we show how our tools may successfully be used in some sample applications.  

But first, a note on completeness in the case of concrete coalgebraic logics. Assume that $\func{L}\colon \BA \to \BA$ has a presentation by operations and equations. By Corollary~\ref{cor:L'haspresentation} we know that its lifting $\func{L}'\colon \var{A}\to \var{A}$ has a presentation by operations and equations as well. Furthermore, we know that one-step completeness of a classical abstract coalgebraic logic $(\func{L},\delta)$ transfers to the lifted logic $(\func{L'},\delta')$. As discussed in the paragraph after Definition~\ref{defin:oneStepCompleteness}, for classical concrete coalgebraic logics $(\func{L},\delta)$, it is known that one-step completeness implies completeness with respect to the semantics determined by the initial algebra \cite{KupkeKurzPattinson2004,KurzPetrisan2010}. The proof of this fact (\emph{e.g.}, as presented in \cite[Theorem 6.15]{KurzPetrisan2010}) can be easily adapted to work for $(\func{L}',\delta')$ as well, after one notes that $\func{L}'$ preserves monomorphisms (because $\func{L}$ preserves monomorphisms, $\Skel\func{L}' \cong \func{L}\Skel$ and $\Skel$ preserves and reflects monomorphisms). For semi-primal $\FLew$-algebras, a similar result has been shown in \cite{LinLiau2022}.  
\begin{cor}\label{cor:OneStepImpliesCompleteness}
Let $(\func{L}',\delta')$ be a concrete coalgebraic logic for $\func{T'}\colon \Set \to \Set$ such that $\func{L}'$ preserves monomorphisms. Then one-step completeness implies completeness. In particular, this holds if $(\func{L}',\delta')$ is a lifting of a classical coalgebraic logic as in Theorem~\ref{thm:OneStepComplPreservation}.   
\end{cor}
This means that, as soon as we find a presentation of the lifted functor $\func{L}'$ occurring in the lifting $(\func{L}',\delta')$ of the classical concrete logic $(\func{L},\delta)$, we get the many-valued completeness result directly from the corresponding classical one. In the following, we show that it is sometimes possible to come up with a presentation of $\func{L}'$ in a straightforward way. For this, we make use of the algebraic structure of $\alg{D}$ (recall Assumption~\ref{ass:mainassumption}). 

\subsection{Lifting Axiomatizations}\label{subs:LiftingAxiomatizations}
Similar to \cite[Section 4]{KurzPoiger2023}, to lift axiomatizations we need to take the algebraic structure of $\alg{D}$ into consideration. In particular, we make use of characterization (5) in Theorem~\ref{thm:semiprimal_characterizations} to identify elements $e\in \alg{D}$ from the information $\tau_d(e)$, where $d$ ranges over the set $D{\setminus}\{ 0 \}$, which we will denote by $D^+$. 

Note that the map $\tau_{(\cdot)}(e)\colon D^+ \to 2$ satisfies $\tau_{(d_1 \vee d_2)}(e) = \tau_{d_1}(e) \wedge \tau_{d_2}(e)$ since $e \geq d_1 \vee d_2 \Leftrightarrow e \geq d_1$ and $e \geq d_2$. Conversely, every map $D^+ \to 2$ satisfying this condition is of this form. This is proved completely analogous to \cite[Lemma 14]{KurzPoiger2023}).         
\begin{lem}\label{lem:tausCharacterization}
Let $\mathcal{T}\colon D^+ \to 2$ satisfy $$\mathcal{T}(d_1 \vee d_2) = \mathcal{T}(d_1) \wedge \mathcal{T}(d_2)$$ for all $d_1, d_2 \in D^+$. Then $\mathcal{T} = \tau_{(\cdot)}(e)$ for $e = \bigvee\{ d \mid \mathcal{T}(d) = 1 \}$.
\end{lem}

Suppose that $\func{L}\colon \BA \to \BA$ has a presentation by one unary operation $\Box$. We want to find a presentation of $\func{L}'\colon \var{A} \to \var{A}$ by one unary operation $\Box'$ as well. In light of Lemma~\ref{lem:tausCharacterization}, the idea now is to try to define it in such a way that $\Box'a$ is pieced together by the (crisp) information of all $\Box \tau_d(a)$. We show that this works if the original $\Box$ is meet-preserving. We thus generalize \cite[Corollary 16]{KurzPoiger2023} from the primal to the semi-primal case. 

\begin{thm}\label{thm:liftPresentationTaus}
Let $\func{L}\colon \BA \to \BA$ have a presentation by one unary operation $\Box$ and equations which all hold in $\alg{D}$ if $\Box$ is replaced by any $\tau_d, d\in D^+$, including the equation $\Box(x\wedge y) = \Box x \wedge \Box y$.
Then $\func{L}'$ has a presentation by one unary operation $\Box'$ and the following equations.
\begin{itemize}
\item $\Box'$ satisfies all equations which the original $\Box$ satisfies,
\item $\Box' \tau_d(x) = \tau_d(\Box' x)$ for all $d\in D^+$.
\end{itemize}
\end{thm}

\begin{proof}
Let $\func{L^\tau}\colon \var{A} \to \var{A}$ be the endofunctor presented by the operation $\Box'$ and the equations from the statement. Since $\func{L^\tau}$ and $\func{L}'$ are both finitary and $\func{P',S'}$ restrict to a dual equivalence between the full subcategories $\var{A}^\mathrm{fin}$ and $\SetD^\mathrm{fin}$ consisting of the corresponding finite members, it suffices to show
$$
\func{S'L^\tau P'} \cong \func{S' L' P'}
$$
on the finite level, and from now on we only consider the restrictions of functors to this finite level. Let $\func{T'}$ and $\func{T}$ be the duals of $\func{L'}$ and $\func{L}$, respectively. Since $\func{S'L'P' \cong T'}$ holds, we can equivalently show $\func{S'L^\tau P'} \cong \func{T'}$. By Definition~\ref{def:canonicalLiftingSetD} and Proposition~\ref{prop:LiftingSetStandard}, the functor $\func{T}'$ is completely characterized by $\CS\func{T'}\cong \func{T}\CS$ for all $\alg{S}\in \mathbb{S}(\alg{D})$ (note that, in particular this includes $\U \func{T'} \cong \func{T}\U$). Thus, altogether it suffices to show 
$$
\CS\func{S' L^\tau P'} \cong \func{SLP}\CS \text { for all } \alg{S}\in \mathbb{S}(\alg{D}).
$$ 
By definition of the functors involved, we want to find a bijection between the sets $\var{A}\big(\func{L^\tau} (\prod_{X}v(x)),\alg{S}\big)$ and $\BA\big(\func{L}(\alg{2}^{\CS(X,v)}),\alg{2}\big)$ which is natural in $(X,v)\in \SetD^\mathrm{fin}$. By definition of $\func{L}^\tau$, the former set is naturally isomorphic to the collection of all functions 
$$
f\colon \{ \Box' a \mid a\in \prod_{x\in X} v(x) \} \to \alg{S} \text{ satisfying the equations of } \func{L}^\tau.
$$
Similarly, the latter set is naturally isomorphic to the collection of all functions
$$
g\colon \{ \Box b \mid b\in \alg{2}^{\CS(X,v)} \} \to \alg{2} \text{ satisfying the equations of } \func{L}.
$$
Given $f$ as above, we assign to it $g_f$ defined by $g_f(\Box b) = f(\Box' b^0)$, where $b^0(x) = b(x)$ for $x\in \CS(X,v)$ and $b^0(x) = 0$ otherwise. The map $g_f$ is well-defined since $T_1(f(\Box' b^0)) = f(\Box' T_1(b^0)) = f(\Box' b^0)$, so $f(\Box' b^0) \in \{ 0,1 \}$. Furthermore, $g_f$ satisfies the equations of $\func{L}$ because they are included in the equations of $\func{L}^\tau$, which $f$ satisfies. 

Conversely, given $g$ as above we assign to it $f_g$ defined by 
$$
f_g(\Box' a) = \bigvee\{ d \mid g\big(\Box \tau_d(a^\flat)\big) = 1 \},
$$
where $a^\flat$ is the restriction of $a$ to $\CS(X,v)$. Since we have 
$$
g\big(\Box \tau_{(d_1 \vee d_2)}(a^\flat)\big) = g\big(\Box (\tau_{d_1}(a^\flat)\wedge \tau_{d_2}(a^\flat))\big) = g\big(\Box \tau_{d_1}(a^\flat)\big) \wedge g\big(\Box \tau_{d_2}(a^\flat)\big),
$$
the condition of Lemma~\ref{lem:tausCharacterization} is satisfied here. Therefore, $\tau_d (f_g(\Box' a)) = g(\Box \tau_d(a^\flat))$. On the other hand, we use $\tau_c \circ \tau_d = \tau_d$ and compute 
$$ f_g\big(\Box' \tau_d(a)\big) = \bigvee\{ c \mid g\big(\Box \tau_c (\tau_d(a^\flat))\big) = 1\} = \bigvee\{ c \mid g\big(\Box \tau_d(a^\flat)\big) = 1\} = g\big(\Box \tau_d(a^\flat)\big).$$
Thus we showed that $f_g$ satisfies the equations $\Box' (\tau_d(x)) = \tau_d(\Box x)$ for all $d$. The reason that $f_g$ satisfies the remaining equations of $\func{L}^\tau$, \emph{i.e.}, the equations of $\func{L}$, is that these equations are satisfied by all $\tau_d$ and preserved by $g$. For example, we see that $f_g$ preserves meets by computing
$$
f_g\big(\Box' (a_1 \wedge a_2)\big) = \bigvee \{ d \mid g\big(\Box \tau_d(a_1^\flat \wedge a_2^\flat)\big) = 1 \}
$$
and thus $\tau_d (f_g (\Box' (a_1 \wedge a_2))) = g(\Box \tau_d(a_1^\flat \wedge a_2^\flat))$, which is equal to $g(\Box \tau_d(a_1^\flat)) \wedge g(\Box \tau_d(a_2^\flat))$ because both $\tau_d$ and $g$ preserve meets. On the other hand, we compute 
$$
\tau_d\big(f_g(\Box' a_1) \wedge f_g(\Box' a_2)\big) = \tau_d\big(f_g(\Box' a_1)\big) \wedge \tau_d\big(f_g(\Box' a_2)\big) = g\big(\Box \tau_d(a_1^\flat)\big) \wedge g\big(\Box \tau_d(a_2^\flat)\big).
$$
Thus we showed that $\tau_d (f_g (\Box' (a_1 \wedge a_2))) = \tau_d(f_g(\Box' a_1) \wedge f_g(\Box' a_2))$ holds for all $d\in D^+$, which implies $f_g (\Box' (a_1 \wedge a_2)) = f_g(\Box' a_1) \wedge f_g(\Box' a_2)$ as desired. 

Naturality of the bijection $g \mapsto g_{(\cdot)}$ is easy to check by definition, so we are left to show that the two assignments $f \mapsto g_f$ and $g \mapsto f_g$ are mutually inverse. To show $g_{f_g} = g$, we simply compute
$$
g_{f_g}(\Box b) = f_g(\Box' b^0) = \bigvee\{ d \mid g\big(\Box \tau_d((b^0)^\flat)\big) = 1 \} = \bigvee \{ d \mid g(\Box b) = 1 \} = g(\Box b),
$$
where we used $(b^0)^\flat = b$ which is clear by the definitions and $\tau_d(b) = b$ because $b \in 2^{\CS(X,v)}$. 
Showing that $f_{g_f} = f$ is more involved. We first compute  
\begin{align*}
f_{g_f}(\Box' a) & = \bigvee\{ d \mid g_f\big(\Box \tau_d (a^\flat)\big) = 1\} \\ 
 & = \bigvee\{ d \mid f\big(\Box' \tau_d((a^\flat)^0)\big) = 1 \} \\
 & = \bigvee\{ d \mid \tau_d\big(f(\Box' (a^\flat)^0\big) = 1 \} = f\big(\Box' (a^\flat)^0\big).
\end{align*}  
This means we have to show that $f(\Box' a) = f(\Box' \tilde{a})$ always holds for $\tilde{a}(x) = a(x)$ on $\CS(X,v)$ and $\tilde{a}(x) = 0$ on $X{\setminus}\CS(X,v)$. Clearly this holds if $f$ is constant, so assume that $f$ is not constant. It suffices to show $f(\Box' \alpha) = 1$ for 
$$
\alpha (x) =  \begin{cases}
1 & \text{ if } x \in \CS(X,v) \\
0 & \text{ if } x \notin \CS(X,v),
\end{cases}
$$  
because this implies $f(\Box' \tilde{a}) = f(\Box'(a \wedge \alpha)) = f(\Box' a) \wedge f(\Box' \alpha) = f(\Box' a).$ In order to show that $f(\Box' \alpha) = 1$, we show that $f(\Box' c_x) = 1$ holds for all $x\in X{\setminus}\CS(X,v)$, where
$$
c_x (y) =  \begin{cases}
1 & \text{ if } y \neq x \\
0 & \text{ if } y = x,
\end{cases}
$$
which is sufficient because $\alpha = \bigwedge\{ c_x \mid x\in X{\setminus}\CS(X,v) \}$ (note that this is a finite meet because $X$ is finite). So let $x\in X{\setminus}\CS(X,v)$ and choose some $d\in D{\setminus}v(x)$. Let $s$ be the minimal element of $\alg{S}$ strictly above $d$ (which exists because $d \neq 1$). Let $c_x^d$ be defined by
$$
c_x^d (y) =  \begin{cases}
1 & \text{ if } y \neq x \\
d & \text{ if } y = x,
\end{cases}.
$$
Now $\tau_d(f(\Box' c_x^d)) = f(\Box' \tau_d(c_x^d)) = f(\Box' 1) = 1$ (since otherwise $f(\Box' 1) = 0$ and the fact that $f$ is order-preserving would imply that $f$ is constant $0$). Since $f(\Box' c_x^d) \in \alg{S}$ and $f(\Box' c_x^d) \geq d$, due to our choice of $s$ we have $f(\Box' c_x^d) \geq s$ as well. This implies 
$$
1 = \tau_s \big(f(\Box' c_x^d)\big) = f\big(\Box' \tau_s(c_x^d)\big) = f(\Box' c_x)
$$
as desired, finishing the proof.        
\end{proof}

With the presentation of this theorem, the natural transformation $\delta' \colon \func{L'P'} \Rightarrow \func{P'T'}$ can be obtained from $\delta$ componentwise via 
\begin{align*} 
\delta'_{(X,v)} \colon \func{L}'\big(\prod_{x\in X} v(x)\big) & \to \prod_{Z\in \func{T}X}\hat{v}(Z) \\ 
\Box' a & \mapsto \big(Z \mapsto \bigvee\{ d \mid \delta\big(\Box \tau_d(a) \big)(Z) = 1 \}\big),
\end{align*}
This means that in this case we have complete description of a lifted concrete coalgebraic logic $(\func{L}',\delta')$.

The applicability of Theorem~\ref{thm:liftPresentationTaus} depends on the specific choice of presentation of $\func{L}$. For example, while it does apply to the presentation given by $\Box 1 = 1$ and $\Box (x\wedge y) = \Box x \wedge \Box y$, it does not apply to the (naturally isomorphic) functor presented by one unary operation $\Diamond$ with equations $\Diamond 0 = 0$ and $\Diamond (x \vee y) = \Diamond x \vee \Diamond y$. However, not surprisingly, in this example an order-dual version of Theorem~\ref{thm:liftPresentationTaus} applies. Let $D^- := D{\setminus}\{ 1 \}$ and for all $d\in D^-$ define 
$$
\kappa_d (x) =  \begin{cases}
0 & \text{ if } x \leq d \\
1 & \text{ if } x \not\leq d.
\end{cases}
$$        
Then the following can then be proved completely analogous to Theorem~\ref{thm:liftPresentationTaus}. 
\begin{thm}\label{thm:liftPresentationEtas}
Let $\func{L}\colon \BA \to \BA$ preserve mono- and epimorphisms and have a presentation by one unary operation $\Diamond$ and equations which all hold in $\alg{D}$ if $\Diamond$ is replaced by any $\kappa_d, d\in D^-$, including the equation $\Diamond(x\vee y) = \Diamond x \vee \Diamond y$.
Then $\func{L}'$ has a presentation by one unary operation $\Diamond'$ and the following equations.
\begin{itemize}
\item $\Diamond'$ satisfies all equations which the original $\Diamond$ satisfies,
\item $\Diamond' \kappa_d(x) = \kappa_d(\Diamond' x)$ for all $d\in D^-$.
\end{itemize}
\end{thm}

Of course, Theorems~\ref{thm:liftPresentationTaus} and \ref{thm:liftPresentationEtas} do not exhaust all possible presentations a functor $\func{L}$ may have, and as of yet we know no systematic method to lift presentations that don't fall within the scope of these theorems. Nevertheless, these theorems already cover some ground, most notably including classical modal logic. In the following subsection, we show how our results apply, among others, to this case.  

\subsection{Sample Applications}\label{subsec:applications} 
We now have a considerable toolkit to lift classical coalgebraic logics to the semi-primal level. In this subsection, we give some sample applications of these tools and indicate how they relate to prior results in many-valued modal logic. As our first example, we apply our techniques to classical modal logic.
\subsubsection{Classical Modal Logic}\label{subsubs:classicalModalLogic}
As in Example~\ref{exam:classicalModalLogic}, let $\func{T} = \mathcal{P}$ be the powerset functor and $\func{L}\colon \BA \to \BA$ be the functor presented by a unary operation $\Box$ and equations $\Box 1 = 1$ and $\Box(x \wedge y) = \Box x \wedge \Box y$. The natural transformation $\delta\colon \func{LP} \Rightarrow \func{PT}$ is given on $X\in \Set$ by
$$
\Box Y \mapsto \{ Z \subseteq X \mid Z \subseteq Y \} \text{ for } Y\subseteq X.
$$
It is well-known that the classical coalgebraic logic $(\func{L},\delta)$ thus defined is (one-step) complete for $\mathcal{P}$ and expressive if we replace $\mathcal{P}$ by the finite powerset functor $\mathcal{P}_\omega$.

Now let $\mathcal{P}'$ be the lifting of $\mathcal{P}$ to $\SetD$. As described in Example~\ref{exam:LiftingPowerset}, the corresponding coalgebras are the Kripke $\alg{D}$-frames $(X,R,v)$ where $(X,v)\in \SetD$ is compatible with the accessibility relation in the sense that $x_1 R x_2 \Rightarrow v(x_2) \leq v(x_1)$.  

Let $(\func{L'},\delta')$ be the lifting of $(\func{L},\delta)$ to $\var{A}$. By Theorem~\ref{thm:liftPresentationTaus} we know that $\func{L}'$ has a presentation by one unary operation $\Box'$ and equations 
\begin{align*} \Box' 1 = 1, && \Box' (x \wedge y) = \Box' x \wedge \Box' y, && \Box'\tau_d(x) = \tau_d (\Box' x), \text{ for all } d\in D^+.
\end{align*}                  
The natural transformation $\delta' \colon \func{L'P'} \Rightarrow \func{P'}\mathcal{P}'$ has components 
$$
\delta'_{(X,v)}(\Box' a)(Z) = \bigvee\{ d \mid \delta_X\big(\Box \tau_d(a)\big)(Z) = 1\}.
$$
Due to the chain of equivalences 
\begin{align*}
\delta_X\big(\Box \tau_d(a)\big)(Z) = 1 & \Leftrightarrow \forall z\in Z\colon \tau_d(a)(z) = 1 \\ & \Leftrightarrow \forall z\in Z\colon a(z) \geq d \\ & \Leftrightarrow \bigwedge_{z\in Z}a(z) \geq d \Leftrightarrow \tau_d\big(\bigwedge_{z\in Z} a(z)\big) = 1,
\end{align*}
we can simplify this to 
$$
\bigvee\{ d \mid \delta_X\big(\Box \tau_d(a)\big)(Z) = 1\} = \bigvee \{ d \mid \tau_d\big(\bigwedge_{z\in Z} a(z)\big) = 1 \} = \bigwedge_{z\in Z} a(z).
$$
Thus, $\delta'$ yields the conventional semantics of a many-valued box-modality via meets. 

From Theorem~\ref{thm:OneStepComplPreservation} and Corollary~\ref{cor:ReplaceT'byT} we get that $(\func{L',\delta'})$ and $(\func{L'},\delta^\top)$ are one-step complete since $(\func{L},\delta)$ is. By Corollary~\ref{cor:OneStepImpliesCompleteness} this implies completeness for Kripke ($\alg{D}$-)frames. Such completeness results has been proven directly in the case where $\alg{D} = \lucas_n$ is a finite \L ukasiewicz chain in \cite{Bou2011, HansoulTeheux2013} and (only for $(\func{L}', \delta^\top)$) in the case where $\alg{D}$ is a Heyting algebra expanded with the unary operations $T_d$ from Theorem~\ref{thm:semiprimal_characterizations} in \cite{Maruyama2009}. 

Replacing $\mathcal{P}$ by $\mathcal{P}_\omega$, from Theorem~\ref{thm:ExpressPreservation} we get that $(\func{L}',\delta')$ and $(\func{L}',\delta^\top)$ are expressive, that is, they satisfy the Hennessy-Milner property. In the case where $\alg{D} = \lucas_n$ with semantics $\delta^\top$ this has been shown in \cite{MartiMetcalfe2018} (see also \cite{BilkovaDostal2016} for a coalgebraic treatment via predicate liftings).

Of course, it is also possible to use Theorem~\ref{thm:liftPresentationEtas} instead to get the analogous completeness and expressivity result for the many-valued $\Diamond'$ satisfying $\Diamond' 0 = 0$, $\Diamond'(x \vee y) = \Diamond' x \vee \Diamond' y$ and $\Diamond' \kappa_d(x) = \kappa_d(\Diamond' x)$ for all $d\in D^-$. Here, as usual, a formula $\Diamond' \varphi$ is interpreted on Kripke ($\alg{D}$-)frames as a join. 
\subsubsection{Filter Frames}\label{subsubs:MonotoneNeighborhoodFrames}
A neighborhood frame $(X,N)$ (see Example~\ref{exam:NeighborhoodFrames}) is a \emph{filter frame} if, for all $x\in X$, the collection of neighborhoods $N(x)$ is closed under finite intersections and supersets (note that we allow $N(x)$ to be empty). The \emph{filter functor} $\var{M}\colon \Set \to \Set$ is the corresponding subfunctor of the neighborhood functor $\var{N}$.

Let $\func{L}\colon \BA \to \BA$ be presented by one unary operation $\Box$ and the equation $\Box(x \wedge y) = \Box x \wedge \Box y$ and let $\delta$ be defined as in Example~\ref{exam:classicalNeighborhoodLogic}. Then the concrete coalgebraic logic $(\func{L},\delta)$ for $\var{M}$ is well-known to be complete. It is expressive if $\var{M}$ is replaced by the functor $\var{M}_\omega$ corresponding to (image-)finite filter frames \cite{HansenKupkePacuit2009}. 

Let $(\func{L}', \delta')$ be the lifting of $(\func{L},\delta)$ to $\var{A}$. Again, by Theorem~\ref{thm:liftPresentationTaus} we get a presentation of $\func{L}'$ by one unary operation $\Box'$ and equations 
\begin{align*}
\Box' (x \wedge y) = \Box' x \wedge \Box' y, && \Box'\tau_d(a) = \tau_d(\Box a) \text{ for all } d\in D^+. 
\end{align*}     
The corresponding semantics $\delta'_{(X,v)} \colon \func{L'P'}(X,v) \to \func{P'}\var{M}'(X,v)$ are given by 
$$
\delta'_{(X,v)}(\Box' a)(N) = \bigvee\{ d \mid \delta_X\big(\Box \tau_d(a)\big)(N) = 1\} = \bigvee \{d \mid \tau_d(a) \in N \}.
$$ 
This can be interpreted as follows. Given a neighborhood model $(X,N,\mathsf{Val})$ where $\mathsf{Val}\colon X \times \mathsf{Prop} \to \alg{D}$, for a formula $\varphi \in \alg{D}^X$ we have
$$
\mathsf{Val}(x,\Box'\varphi) = \bigvee\{ d \mid \tau_d(\varphi) \in N(x) \}.
$$ 
By Theorem~\ref{thm:OneStepComplPreservation} and Corollary~\ref{cor:OneStepImpliesCompleteness} we know that $(\func{L}', \delta')$ is complete for $\mathcal{M}'$- and $\mathcal{M}$-coalgebras. Replacing $\var{M}$ by $\var{M}_\omega$ we also get the corresponding expressivity results by Theorem~\ref{thm:ExpressPreservation}. 

\subsubsection{Neighborhood Frames}\label{subsubs:NeighborhoodFrames}
We conclude with an example where the theory of Section~\ref{sec:ManyvaluedCoalgModalLogic} of lifting abstract coalgebraic logics applies but obtaining an axiomatization via Theorem~\ref{thm:liftPresentationTaus} or \ref{thm:liftPresentationEtas} is not possible. 

Let $\func{L}\colon \BA \to \BA$ be the functor which has a presentation by one unary operation $\bigtriangleup$ and no (that is, the empty set of) equations. Let $\delta$ be as in Example~\ref{exam:NeighborhoodFrames} again. The concrete coalgebraic logic $(\func{L},\delta)$ for $\var{N}$ is again complete, and expressive if we replace $\var{N}$ by an appropriate $\var{N}_\omega$. 

Therefore, the lifting $(\func{L}',\delta')$ of $(\func{L},\delta)$ to $\var{A}$ is a complete (resp. expressive) abstract coalgebraic logic for $\var{N}'$ (resp. $\var{N}'_\omega$).

Furthermore, by Theorem~\ref{cor:L'haspresentation} we know that $\func{L}'$ does have a presentation by operations and equations. However, since the theory of Subsection~\ref{subs:LiftingAxiomatizations} doesn't apply to $\func{L}$, as of yet we do not know a concrete presentation of $\func{L}'$ in the case where $\alg{D}$ is semi-primal but not primal (for the primal case, a presentation of $\func{L}'$ by $|D^+|$-many operations $\bigtriangleup_d$ is provided in \cite[Theorem 15]{KurzPoiger2023}). 

\section{Conclusion}\label{sec:conclusion}
We provided a general method to lift classical abstract coalgebraic logics $(\func{L},\delta)$ to many-valued abstract coalgebraic logics $(\func{L}',\delta')$, where $\func{L}'$ is defined on $\var{A}$, the variety generated by a semi-primal lattice-expansion $\alg{D}$. We also showed that $(\func{L}',\delta')$ inherits the properties of one-step completeness and expressivity from $(\func{L},\delta)$. Furthermore, $\func{L'}$ has a presentation by operations and equations if $\func{L}$ has one, and we partially answered how such a presentation of $\func{L}'$ can be obtained directly from a presentation of $\func{L}$. In particular, we showed how to apply these methods to lift classical modal logic. In the following, we offer some possible directions to follow up this research.

The results of this paper are heavily dependent on the fact that $\alg{D}$ is semi-primal, as this gives rise to the nicely structured dual category $\StoneD$. However, there are many other natural dualities \cite{ClarkDavey1998} on which coalgebraic logics could be built. In particular, an obvious generalization would be to consider algebras $\alg{D}$ which are \emph{quasi-primal} (defined in the paragraph after Definition~\ref{def: semi-primal-algebra}). Another important extension of the results of this paper would be to cover some \emph{infinite} algebras of truth-degrees as well. Future work will be dedicated to attempting this for the \emph{standard $\mathsf{MV}$-chain} based on the real unit interval, providing a coalgebraic analogue of \cite{HansoulTeheux2013}.      

Theorems~\ref{thm:liftPresentationTaus} and \ref{thm:liftPresentationEtas} only deal with presentations of $\func{L}$ by a single unary operation, though the proof allows a straightforward generalization to presentations by a single $n$-ary operation. The case of presentations by more than one operation seems more involved. For example, is it possible to lift the presentation of a complete logic for the distribution functor from \cite{CirsteaPattinson2004} to the semi-primal level?

Other interesting questions arise if we replace the variety of Boolean algebras, on which the `original' two-valued logic is defined, by other ones. For example, building on Priestley duality instead of Stone duality, is it possible to lift positive modal logic (described coalgebraically in \cite{Palmigiano2004}) to a many-valued level? Here, the theory of \emph{order-primality} might come in handy. 

Lastly, in this paper we focused on coalgebraic logics from the point of view of their algebraic semantics similar to \cite{KupkeKurzPattinson2004}, but to gain a more complete picture of many-valued coalgebraic logics we also encourage their further study via relation liftings \cite{Moss1999} and predicate liftings \cite{Pattinson2003}. 

\section*{Acknowledgment}
  \noindent The second author is supported by the Luxembourg National Research Fund under the project  PRIDE17/12246620/GPS.

\bibliographystyle{alphaurl}
\bibliography{References}
\end{document}